\documentclass[lettersize,journal]{IEEEtran}
\usepackage[outdir=./]{epstopdf}
\usepackage{pgfplots}
\pgfplotsset{compat=1.16}
\usepackage{amssymb}
\usepackage{amsmath}
\usepackage[
lined,
boxed,
ruled,
linesnumbered
]{algorithm2e}
\SetKwFor{For}{for}{do}{end}
\SetKwFor{ForPar}{forall}{do in parallel}{end}
\SetKwFor{ForEach}{foreach}{do}{end}
\SetKwIF{If}{ElseIf}{Else}{if}{then}{elseif}{else}{end}
\usepackage{mathrsfs}
\usepackage{bm}
\newcommand{\bs}[1]{{\bm{#1}}}
\usepackage{graphicx,booktabs,multirow}
\usepackage{enumerate} 
\usepackage{epstopdf}
\usepackage{tikz}
\usetikzlibrary{arrows,decorations.markings}
\usepackage{caption}
\usepackage{setspace}
\usepackage{graphics}
\usepackage[labelformat=simple]{subcaption}
\usetikzlibrary{shapes}
\usepackage{amsthm}
\usepackage{cite}
\usepackage{letltxmacro}
\LetLtxMacro{\originaleqref}{\eqref}
\LetLtxMacro{\originalref}{\ref}

\definecolor{colorhkust}{HTML}{142B8C}
\definecolor{colorshanghaitech}{HTML}{A20005}
\definecolor{colortsinghua}{HTML}{743481}
\definecolor{colordark}{RGB}{184,134,11}
\definecolor{colorRed}{RGB}{128, 0, 0}
\definecolor{colorGreen}{RGB}{0, 64, 0}
\definecolor{colorBlue}{RGB}{0, 0, 128}
\usepackage[
colorlinks, 
linkcolor=colorRed, 
anchorcolor=blue, 
citecolor=colorBlue, 
urlcolor = colorGreen
]{hyperref}
\usepackage{etoolbox}

\renewcommand{\eqref}{\originaleqref}

\newcommand{\sat}{\mathsf{S}}
\newcommand{\isl}{\mathsf{ISL}}
\newcommand{\gsl}{\mathsf{GSL}}
\newcommand{\bz}{\bs{z}}

\newcommand{\E}{\mathcal{E}}
\newcommand{\fedmega}{\textsc{FedMega}\xspace}
\newcommand{\fedavg}{\textsc{FedAvg}\xspace}
\newcommand{\fedisl}{\textsc{FedISL}\xspace}
\newcommand{\hlsgd}{\textsc{HL-SGD}\xspace}

\renewcommand{\mod}{\%}
\newtheorem{thm}{Theorem}
\newtheorem{lem}{Lemma}

\newtheorem{ass}{Assumption}
\newtheorem{cor}{Corollary}
\newtheorem{Rem}{Remark}

\tikzset{
	->-/.style={decoration={markings, mark=at position 0.5 with {\arrow{stealth}}}, postaction={decorate}}
}

\begin{document}

	\title
	{Satellite Federated Edge Learning: Architecture Design and Convergence Analysis}
	
        \author{Yuanming~Shi,~\textit{Senior Member, IEEE},~Li~Zeng,~\textit{Graduate Student Member, IEEE},~Jingyang~Zhu,~\textit{Graduate Student Member, IEEE},~Yong~Zhou,~\textit{Senior Member, IEEE},~Chunxiao~Jiang,~\textit{Senior Member, IEEE},\\~and Khaled~B.~Letaief,~\textit{Fellow, IEEE}
\thanks{Yuanming Shi, Li Zeng, Jingyang Zhu, and Yong Zhou are with the School of Information Science and Technology, ShanghaiTech University, Shanghai 201210, China (e-mail: $\{$shiym, zengli, zhujy2, zhouyong$\}$@shanghaitech.edu.cn). 
	
			Chunxiao Jiang is with the Tsinghua Space Center, Tsinghua University, Beijing,
			100084, China. (e-mail: jchx@tsinghua.edu.cn).
                
			Khaled B. Letaief is with the Department of Electronic and Computer Engineering, Hong Kong University of Science and Technology, Clear Water Bay, Hong Kong (e-mail: eekhaled@ust.hk).
}
}
	
	\maketitle
	\IEEEpeerreviewmaketitle
	

\begin{abstract}
    The proliferation of low-earth-orbit (LEO) satellite networks leads to the generation of vast volumes of remote sensing data which is traditionally transferred to the ground server for centralized processing, raising privacy and bandwidth concerns. Federated edge learning (FEEL), as a distributed machine learning approach, has the potential to address these challenges by sharing only model parameters instead of raw data.
    Although promising, the dynamics of LEO networks, characterized by the high mobility of satellites and short ground-to-satellite link (GSL) duration, pose unique challenges for FEEL. Notably, frequent model transmission between the satellites and ground incurs prolonged waiting time and large transmission latency. This paper introduces a novel FEEL algorithm, named \fedmega, tailored to LEO mega-constellation networks. By integrating inter-satellite links (ISL) for intra-orbit model aggregation, the proposed algorithm significantly reduces the usage of low data-rate and intermittent 
    GSL. Our proposed method includes a ring all-reduce based intra-orbit aggregation mechanism, coupled with a network flow-based transmission scheme for global model aggregation, which enhances transmission efficiency. Theoretical convergence analysis is provided to characterize the algorithm performance. Extensive simulations show that our \fedmega algorithm outperforms existing satellite FEEL algorithms, exhibiting an approximate 30\% improvement in convergence rate.
\end{abstract}

	\section{Introduction}

As the 5G communication technology has matured into its commercialization phase and witnessed widespread global deployments \cite{shafi20175g}, both the academic and industrial communities have turned their attention towards the conceptualization and foundational groundwork for the impending 6G communication systems \cite{shi2023task}.
The seamless integration of artificial intelligence (AI) with communication networks, as well as the emergence of non-terrestrial networks, are widely considered as the most salient features of the upcoming 6G era \cite{EdgeAI6G, giordani2020non}. Thanks to the abundance of data, advancements in machine learning (ML) algorithms, and exponential growth in computing power, AI has shown unparalleled proficiency across diverse tasks, such as image analytics and natural language processing. Recent progresses on large language models, e.g., ChatGPT and Bard, have further unleashed the wave of foundation models including general large language models and diverse industry-specific models such as medical models and remote sensing models \cite{bommasani2021opportunities}.
Meanwhile, low-earth-orbit (LEO) satellite mega-constellations have been in active deployment around the world, e.g., Starlink \cite{starlink2023}, Kuiper \cite{kuiper2023}, as well as China's recently announced LEO constellation project for integrated communication and remote sensing \cite{spacequip2023}. Modern LEO mega-constellations typically encompass hundreds to tens of thousands of satellites flying at altitudes ranging from $300$ to $2000$ km, with an objective to achieve ubiquitous coverage in space, air, ground and sea, and further provide an array of services from communication and remote sensing to navigation and computation \cite{xie2021leo}. 

In addition to communication capabilities, modern LEO satellites are often equipped with various payloads such as computation and remote sensing. Recently, there has been a rapid development in various satellite on-board computing hardwares, such as on-board CPU \cite{geist2019spacecube}, GPU \cite{kosmidis2020gpu4s}, and other emerging in-orbit computing chips, leading to consistent advancements in on-board computation capabilities. These progresses make the in-orbit edge computing a reality \cite{bhattacherjee2020orbit} and on-board ML a promising trend \cite{izzo2022selected, zhang2022progress}. Meanwhile, benefiting from the low altitude of LEO satellites and maturity of remote sensing techniques like Synthetic Aperture Radar (SAR) and multi-spectral sensing technologies, coupled with the explosive growth in the number of satellites, an immense volume of high-resolution remote sensing data is being incessantly collected at the satellite edge \cite{guo2016big}. Behind these voluminous and diverse data lies significant potential for information extraction and utilization. ML techniques, as the most prominent data mining methods in the current AI era, have been extensively explored in remote sensing data processing and offer various applications ranging from land cover classification \cite{tong2020land}, cloud detection \cite{li2020accurate}, and precipitation estimation \cite{chen2019machine}, to oil spill detection \cite{diana2021oil}, and wildfire detection \cite{hawkins2022creating}, etc. 

The traditional way of training ML models on such massive remote sensing data involves transmitting the raw data to ground stations (GS) and further processing these data in cloud computing centers \cite{Li2021}. However, with the dramatic increase in data volumes and the limited and expensive spectral resources between the satellites and the ground, this method incurs unbearable and increasing communication costs. In addition, the utilization of raw data involves severe data privacy issues, especially for high-resolution satellite images, the use of which is stringently regulated by relevant legislation \cite{maniadaki2021reconciling, coffer2020balancing}. In light of the two challenges on communication bottlenecks and data privacy, employing the computing capabilities at the satellite edge for in-orbit ML model training is emerging as a major trend \cite{izzo2022selected, ruuvzivcka2023fast}. More specifically, we propose to leverage the federated edge learning (FEEL) paradigm \cite{tak2020federated} for satellites to tackle the challenge of data mining on massive distributed satellite data. In FEEL, we aim to collaboratively train a global model with data from multiple edge devices without any raw data leaving their devices. Specifically, the edge devices first employ various optimization methods such as stochastic gradient descent (SGD) to train a local model. The edge devices then transmit the local model parameters to the parameter server (PS) which aggregates the received model parameters and broadcast them back to all the edge devices, creating a iterative model training process. In FEEL, only model parameters are shared instead of massive raw data. This not only mitigates substantial communication overheads but also addresses potential privacy concerns associated with raw data sharing, effectively tackling the aforementioned two challenges.

\subsection{Challenges and Related Works}
\subsubsection{Challenges for On-Board FEEL}
Although the deployment of FEEL in terrestrial networks has been extensively studied and the solutions are relatively mature \cite{yang2020federated,wang2021federated,shi2023vertical}, implementing FEEL in LEO mega-constellation networks presents unique challenges due to the high mobility of satellites, i.e., extremely short communication windows and sporadic satellite-ground connections \cite{matthiesen2023federated}. One of the significant characteristics of FEEL is the frequent model exchange between edge devices and the PS. However, in LEO networks, the communication window between a satellite, acting as the edge device, and the GS is exceptionally short. Missing this window necessitates a prolonged waiting period before re-establishing the connection, during which the FEEL process has to be suspended. Consequently, the ground PS often has to wait an extended time to collect all local models from the satellites, resulting in high latency.
\subsubsection{Recent Studies for On-Board FEEL}
Many existing studies focused on proposing scheduling schemes to address challenges posed by ground-to-satellite link (GSL). For instance, the authors in \cite{elmahallawy2022asyncfleo, so2022fedspace,razmi2022ground ,lin2023fedsn} focused on designing asynchronous FEEL mechanisms to reduce waiting time for model aggregation.
The authors in \cite{razmi2022scheduling} optimally scheduled the transmission-receiving time of the model parameters between the satellites and the GS.
However, this approach might result in stale model parameters at the PS, causing a certain performance loss.
Leveraging the mature inter-satellite link (ISL) technology for reducing the waiting time during GSL interruptions is an another main research direction \cite{razmi2022board, elmahallawy2023optimizing,razmi2024onboard}. 
The authors in \cite{razmi2022board} first proposed to use the intra-orbit ISLs to assist the model downloading from satellites to the ground and thus save the waiting time. 
Furthermore, the authors in \cite{elmahallawy2023optimizing} designed a scheme to choose the better sink satellite which acts as the relay node between other satellites and the ground.
Moreover, inter-orbit ISL was also considered for the model aggregation across different orbits \cite{zhai2023fedleo,wu2022DSFL}. 
\subsubsection{Recent Studies for Decentralized FEEL}
Different from the client-server architecture \cite{mcmahan2017communication}, decentralized (gossip) SGD that utilizes the device-to-device communications has been considered in decentralized FEEL. For instance, a general framework for decentralized SGD with dynamic topology is proposed in \cite{koloskova2020unified}.
The authors in \cite{wu2022DSFL} focused on decentralized satellite FEEL without the coordination of a ground PS, which substitutes the aggregation function of the ground PS for model aggregation both within the same satellite orbit and between different satellite orbits via ISLs.
In addition, the authors in \cite{guo2022hybrid} proposed to leverage both device-to-device and device-to-server communications to facilitate an hybrid architecture that combines the client-server and decentralized FEEL architectures to accelerate the model convergence.
\subsubsection{Limitations of Existing Studies}
However, existing works on deploying FEEL in LEO satellite networks with ISLs have not addressed several critical issues. Primarily, the existing literature only focuses on systems with a single GS, while in real systems there are usually multiple GSs communicating with satellites. Utilizing multiple GSs simultaneously can efficiently accelerate the model transmission between space and ground. 
Moreover, given that multiple satellites can concurrently establish links with a single GS, there is an absence of a well-designed cooperative transmission scheme that allows satellites to collaboratively accomplish a transmission task in a short period of time.
Furthermore, existing literature overlooks the unique characteristics inherent to both ISL and GSL when tailoring rapid FEEL algorithms. Notably, ISL offers a significantly enhanced energy efficiency and data rate compared to GSL. This oversight leads to overuse of GSL in existing FEEL algorithms, resulting in excessive energy consumption and latency. By astutely integrating these characteristics and custom-adjusting the usage frequency of various links in FEEL implementation, we can achieve improved performance with respect to the convergence speed of FEEL algorithms.

To this end, we consider the deployment of FEEL in LEO mega-constellation systems with multiple GSs and propose the \fedmega algorithm through a synergistic integration of architecture design and transmission scheme design.
\subsection{Contributions}
In this paper, we develop an efficient FEEL framework tailored for the LEO mega-constellation network. Leveraging the superior data rate and stability of intra-orbit ISL over GSL, we introduce multiple intra-orbit training rounds per global round to reduce the GSL usage. Utilizing the orbit's ring topology, we propose an efficient communication scheme for intra-orbit aggregation based on the ring all-reduce algorithm. For rapid global model aggregation, we also propose an efficient model transmission mechanism based on the network flow algorithm, which further mitigates the delay. Furthermore, to characterize the convergence performance of the proposed algorithm, we theoretically analyze its convergence rate, taking into account non-convex loss functions and non-IID data distributions.

The major contributions are summarized as follows:
\begin{enumerate}
    \item We propose a novel and fast convergent FEEL algorithm tailored for the LEO network. This approach leverages the inherent stability and ultra-high transmission rates of intra-orbit ISL compared to GSL. Specifically, we increase the model aggregation frequency within each orbit, which solely relies on ISL, while reducing communication frequency between orbiting satellites and the ground PS, which relies on GSL. Within each global iteration, multiple intra-orbit cycles occur, wherein each satellite undergoes local training and subsequently attains the intra-orbit averaged model. At the end of each global iteration, a global aggregation transpires, where the ground PS retrieves models from space via GSL. This approach adeptly minimizes GSL dependence, thereby mitigating communication latency.
    \item To quickly accomplish the intra-orbit model aggregation step, we leverage the ring topology within each orbit and the full-duplex capability of laser ISL and propose an efficient ring all-reduce-based intra-orbit model aggregation scheme. This scheme ensures that every satellite within an orbit can acquire the intra-orbit averaged model, with the execution time upper bounded by a constant irrespective of the number of satellites in the orbit.
    \item To accelerate the global model aggregation step, we propose an efficient network flow-based model transmission scheme. This scheme can maximize the number of model parameters received per time slot by the ground PS. Since the total amount of model parameters to transmit is fixed, the latency is thereby minimized.
    \item We provide a comprehensive convergence analysis of the proposed \fedmega algorithm under non-convex settings and non-IID data distribution. We demonstrate that the proposed \fedmega algorithm attains linear speedup in terms of the number of local updates, the number of LEO satellites, and the number of intra-orbit aggregations. 
\end{enumerate}

In realistic LEO network settings, extensive simulations on both synthetic datasets and real datasets are conducted. The results demonstrate that our proposed \fedmega can achieve a better prediction accuracy and faster convergence rate than existing benchmarks, which verifies the superiority of the proposed \fedmega algorithm. 

\subsection{Organization and Notations}
For the remainder, we begin by introducing the system model of satellite FEEL in Section \ref{sec: system model}. Next, we present the proposed \fedmega framework, which includes the algorithm workflow as well as the efficient model transmission scheme in Section \ref{sec:flsys_design}. Then, we analyze the convergence of the proposed algorithm in Section \ref{sec: convergence}. Subsequently, we provide the simulation results in Section \ref{sec: simulations}. Lastly, we summarize the conclusions and the future research directions.

\emph{Notations:} Italic, bold lower-case, and bold upper-case letters represent scalars, column vectors, and matrices, respectively. Operators $(\cdot)^\top$ and $\text{diag}(\cdot)$ denote transpose and diagonal matrix, respectively. Also, the operator $|\cdot|$ indicates the cardinality of a set or the absolute value of a scalar, while $||\cdot||$ denotes the Euclidean norm. The notation $\cdot \% \cdot$ represents the modulo operation, i.e., $m \% n$ calculates the remainder when dividing $m$ by $n$. The notation $\angle(\cdot,\cdot)$ is used to denote the angle between two vectors, i.e., $\angle ({\bs{a}},{\bs{b}})$ represents the angle between vectors $\bs{a}$ and $\bs{b}$. Throughout the paper, all angles are expressed in degrees. The notations $R$, $T$, $E$, $r$, $t$, and
$e$ denote the number of communication rounds, the number of intra-orbit rounds, the number
of local updates, each communication round, each intra-orbit round, and each local update,
respectively.


\section{System Model}\label{sec: system model}
In this section, we first introduce the mathematical model of FEEL over LEO satellites and then provide a detailed characterization of the satellite FEEL system.
\subsection{FEEL over LEO Mega-Constellation Network}

We consider the deployment of FEEL in an LEO mega-constellation network, as shown in Fig. \ref{fig:spacepart-sysmdl9}, wherein a ground PS coordinates satellites from $ M $ circular orbits indexed by $ \mathcal{M} := \left\{1,2,\dots,M\right\} $ with the help of $ G $ GSs indexed by $ \mathcal{G} := \left\{1,2,\dots,G\right\} $ to collaboratively train a global model. These GSs, dispersed across different geographic regions, primarily serve as relays between the ground PS and satellites. 
Specifically, each orbit $ m \in \mathcal{M}$ comprises $ K_m $ satellites denoted as $ \mathcal{K}_m := \left\{\sat_{m,1}, \sat_{m,2}, \dots, \sat_{m,K_m}\right\} $. Each satellite $ \sat_{m,k} $ retains an individual $ D_{m,k} $-sized dataset $ \mathcal{D}_{m,k} = \left\{\left(\bs{x}_{m,k}(i), y_{m,k}(i)\right) \mid i = 1,2,\dots,D_{m,k}\right\} $, where $ \bs{x}_{m,k}(i) $ is the $ i $-th sample's feature and $  y_{m,k}(i) $ is its ground-truth label. The local loss of each satellite $ \sat_{m,k} $ with respect to a $ d $-dimensional model parameter $ \bs{z} \in \mathbb{R}^{d} $ is
\begin{equation}
	F_{m,k}(\bs{z}) := \frac{1}{D_{m,k}} \sum_{i=1}^{D_{m,k}} \ell(\bz;\bs{x}_{m,k}(i), y_{m,k}(i)),
\end{equation}
where $ \ell(\bz;\bs{x}_{m,k}(i), y_{m,k}(i)) $ represents the sample-wise loss function. The primary objective of FEEL is to pursue an optimal $ \bz^{\star} $ that minimizes the global loss function $ F(\bz) $, i.e.,
\begin{equation}
	\bz^{\star} = \arg\min_{\bz\in\mathbb{R}^{d}}  F(\bz) := \sum_{m=1}^{M}\sum_{k=1}^{K_m} w_{m,k}F_{m,k}(\bz),
\end{equation}
where $ w_{m,k} := \frac{D_{m,k}}{\sum_{m=1}^{M}\sum_{k=1}^{K_m}D_{m,k}} $ is the model weight associated with satellite $ \sat_{m,k} $. A common strategy for efficiently solving this problem is to employ distributed Stochastic Gradient Descent (SGD), wherein each satellite updates its local model via SGD and shares its local model with the PS \cite{mcmahan2017communication}, i.e., the ground PS periodically downloads the satellites' local models, aggregates local models to update the global model, and then uploads the updated global model to all satellites to initiate the next round's on-board training, as shown in Fig. \ref{fig:spacepart-sysmdl9}. However, this approach can hardly be directly applied in LEO networks due to the unique challenges. In the following, we introduce the network model of modern LEO mega-constellations and identify these challenges.
\begin{figure}[t]
	\centering
	\includegraphics[width=1\linewidth]{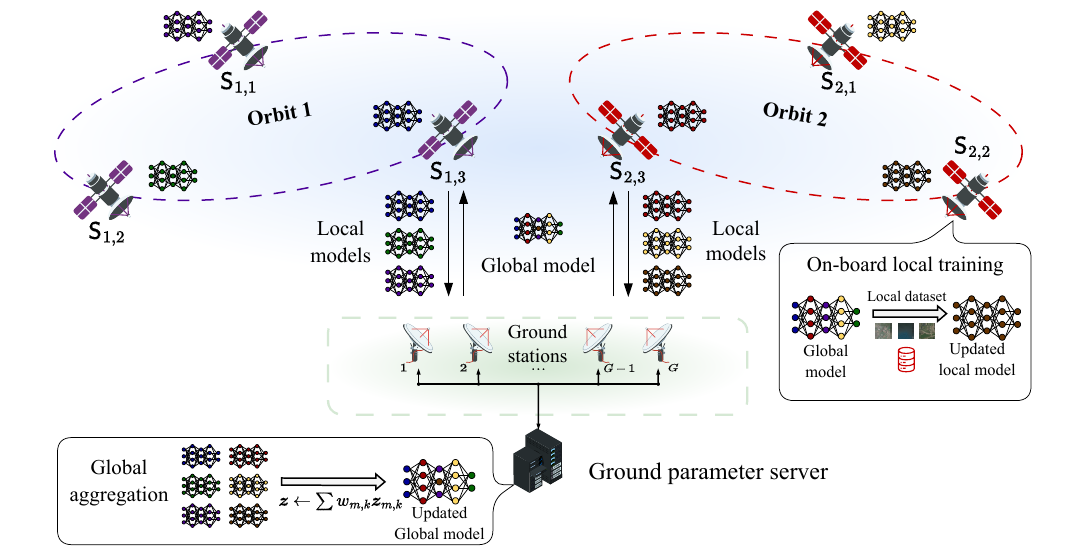}
	\caption[System model of FEEL over satellite CPN]{System model of satellite FEEL.}
	\label{fig:spacepart-sysmdl9}
	\vspace{-0.5cm}
\end{figure}

\subsection{Network Model of LEO Mega-Constellations}
In this subsection, we establish a network model for a LEO mega-constellation, with a focus on describing the network topology as well as the link characteristics. The entire network consists two sub-networks, i.e., the ground sub-network including the ground PS and GSs, and the space sub-network including all satellites. We start by introducing the nodes and connectivity inside each sub-network, followed by a description of the connectivity between two sub-networks. 

\subsubsection{Ground Sub-Network}
The ground sub-network consists of a central PS and $ G $ GSs. The ground PS is the top control unit that is responsible for global model maintenance and coordinating the operations of all nodes. The $ G $ GSs at different geographic locations serve for establishing connectivity between the ground PS and the orbiting satellites. High-speed ground dedicated lines (GDL) are deployed between the PS and each GS \cite{Li2021}.
Note that due to the high construction costs of GSs and GDLs, the number of GSs, i.e., $ G $, is much smaller than the number of satellites in a mega-constellation, i.e., $ G \ll \sum_{m=1}^{M}K_m$. The ground sub-network exhibits a star topology as shown in Fig. \ref{fig:ground_sub_network}.

\begin{figure}[t]
	\centering
	\begin{subfigure}{0.32\linewidth}
		\centering
		\includegraphics[width=0.95\linewidth]{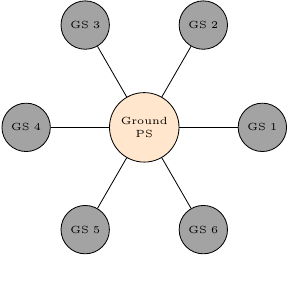}
		\caption{Ground sub-network: Star topology.}
		\label{fig:ground_sub_network}
	\end{subfigure}
	\hfill
	\begin{subfigure}{0.32\linewidth}
		\centering
		\includegraphics[width=0.95\linewidth]{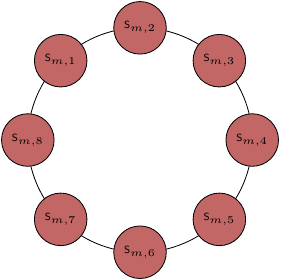}
		\caption{Each orbit in space sub-network: Ring topology.}
		\label{fig:space_sub_network}
	\end{subfigure}
	\hfill
	\begin{subfigure}{0.32\linewidth}
		\centering
		\includegraphics[width=0.95\linewidth]{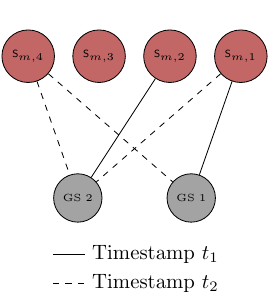}
		\caption{Ground-space connection: Time-varying topology.}
		\label{fig:gs_connection}
	\end{subfigure}
	\caption{Topology structures of network components.}
	\label{fig:topology}
	\vspace{-0.5cm}
\end{figure}

\subsubsection{Space Sub-Network}
The space sub-network is composed of all orbiting satellites, i.e., $ \cup_{m=1}^{M}\mathcal{K}_m $, as well as the ISL. Without loss of generality, we assume that each orbit $ m $ is circular and the $ K_m $ satellites are evenly spaced along the orbit. 
The satellite numbering scheme is as follows: we first arbitrarily choose a satellite and designate it as $ \sat_{m,1} $, then assign the remaining satellites in the ascending order according to their locations along the orbit in the clockwise direction, i.e., each satellite $ \sat_{m,k} \in \mathcal{K}_{m} $ is located between $ 1 $-hop neighbors $ \sat_{m,(k-1) \% K_m} $ and $ \sat_{m,(k+1) \% K_m}  $. 
Since satellites within the same constellation system often have similar hardware configurations, we reasonably assume that the computation and communication hardwares on different satellites share equivalent performance \cite{trevor2023starlink}. 
Typically, the space sub-network exhibits the following features:
\paragraph{Laser-based link implementation and its characteristics}
As laser has been broadly recognized as the primary implementation method for ISL \cite{abdelsadek2022future} and the latest LEO projects, e.g., Telesat, Starlink and Kaskilo, have all started to deploy laser ISLs \cite{erwin2021thales, chaudhry2021laser, toyoshima2020recent}, it is reasonable and practical to incorporate the laser ISL into the system model. Specifically, we outline some of the key characteristics of laser ISL as follows.
\begin{itemize}
	\item \textbf{High data rate and low power consumption.} Laser ISL, operating within the unlicensed spectrum, consistently attains data transmission rates up to 100 gigabits per second, surpassing traditional RF links in both speed and energy efficiency due to its highly focused beams \cite{rob2021why}.
	\item \textbf{Full-duplex capability.} Laser ISLs supports full-duplex communication \cite{heine2010optical}, allowing simultaneous transmitting and receiving.
	\item \textbf{Free of interference.} The focused beam of light used in laser ISL ensures that signals would not be affected by electromagnetic interference or other signal interference.
\end{itemize}
We remark that the above features render laser ISL a particularly attractive option for the transmission of large-sized data.

\paragraph{Stable ring topology inside each orbit}
In each orbit $m \in \mathcal{M}$, each satellite $ \sat_{m,k} \in \mathcal{K}_m $ maintains two laser ISLs to connect its immediately adjacent satellites, i.e., one laser ISL for connecting $ \sat_{m,(k+1) \% K_m} $ and the other for $ \sat_{m,(k-1) \% K_m} $, thus each orbit can be viewed as a ring. 
In dense mega-constellations, the line-of-sight path between adjacent satellites generally is not obstructed by the Earth's surface.
Moreover, due to the fact that satellites in the same orbit are relatively stationary with respect to each other, the intra-orbit laser ISLs are permanently stable \cite{abdelsadek2022future, chaudhry2021laser}.

In addition, inter-orbit laser ISLs are not considered in this work due to issues like severe Doppler shifts deteriorating link quality \cite{kaushal2016optical}. Overall, the space sub-network, as illustrated in Fig. \ref{fig:space_sub_network}, is modeled as a collection of stable rings. For notation ease, $ \isl_{m}(k_1,k_2) $ is used to represent the laser ISL between two satellites $ \sat_{m,k_1} \in \mathcal{K}_m $ and $ \sat_{m,k_2} = \sat_{m,(k_1\pm1)\%K_m} $.
\subsubsection{Ground-Space Connection}
The connections between the ground and space sub-networks rely on the GSL. We use $ \gsl_{g}^{m}(k), g\in\mathcal{G}, m\in\mathcal{M}, k\in\mathcal{K}_m $ to denote the GSL between $ \sat_{m,k} $ and GS $ g $. Typically, GSL exhibits the following features:
\paragraph{RF-based GSL}
Unlike laser ISLs, GSLs are mainly based on radio frequency (RF) technology. Besides, to establish a link with GS $ g $, satellite $ \sat_{m,k} $ must maintain a minimum elevation angle $ \phi_{\mathrm{e}} $, i.e., 
\begin{equation}\label{eqn:gsl_feasibility}
	\angle \left ({\bs{r}}_{\sat_{m,k}} - {\bs{r}}_{g},{\bs{r}}_{g}\right ) \leq 90^\circ - \phi_{\mathrm{e}},
\end{equation}
where $ \bs{r}_{\sat_{m,k}} $ and $ \bs{r}_{g}$ denote the coordinate vectors of $ \sat_{m,k} $ and GS $ g $ in the geocentric coordinate system, respectively. This is mainly applied for preventing severe signal attenuation caused by thick atmosphere.
In addition, GSL in LEO networks exhibits very different characteristics compared to the terrestrial links and laser ISL presented as follows:
\begin{itemize}
	\item \textbf{Low data rate and high power consumption.}  
        Due to a variety of factors such as limited bandwidth, strong interference, severe Doppler effect, significant atmospheric attenuation, and even unstable weather conditions, GSL typically offers a much lower data rate but significantly higher energy cost than laser ISL \cite{cao2022network, kaushal2016optical}. 
	\item \textbf{Short link duration.} Due to the low altitude and high mobility of LEO satellites, the GSL establishment condition \eqref{eqn:gsl_feasibility} can rarely be satisfied. Thus,  each GSL's duration would be extremely short. In modern LEO constellations, the duration of a GSL is only around $ 0.5\sim5 $ minutes \cite{al2021session}. This induces a highly dynamic topology between the ground and space sub-networks, which complicates routing and link scheduling for data transmission \cite{al2021session}.
\end{itemize}

\paragraph{Multi-point to multi-point communication}
With the help of mature multi-beam and multiple access technologies, once the link establishment conditions are satisfied, a single satellite can communicate with multiple GSs at the same time, while a single GS can also serve multiple satellites simultaneously, as shown in Fig. \ref{fig:gs_connection}.

{
	\paragraph{Sporadic ground-space connections}
	Due to the short duration and limited number of GSs, the connectivity between the GSs and satellites tends to display sporadic characteristics \cite{matthiesen2023federated, so2022fedspace}. That is, at any given timestamp, only a small proportion or possibly none of the satellites are capable of establishing feasible GSLs.
}

\begin{figure}[tbp]
	\centering
	\begin{subfigure}{0.9\linewidth}
		\centering
		\includegraphics[width=1\linewidth]{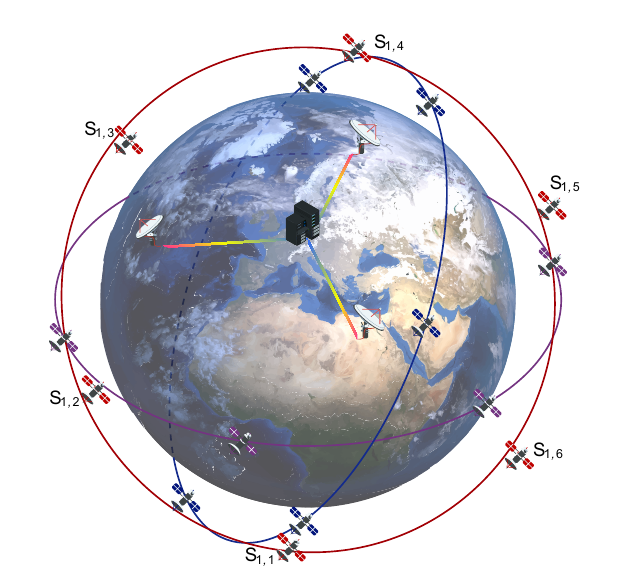}
		\label{fig:systemmodel1}
	\end{subfigure}
	\caption{Overview of the LEO satellite network. Satellites travel in various orbits around the Earth. A ground PS establishes connections with multiple GSs located at different geographical locations.}
	\label{fig:sysmdl}
	\vspace{-0.5cm}
\end{figure}

In summary, the ground sub-network with a star topology and the space sub-network with a multi-ring topology both exhibit stable topological structures. However, the ground-space connection exhibits sporadic feature, highly dynamic topology, and low data rate. The whole system model is depicted in Fig. \ref{fig:sysmdl}. 

\subsection{Challenges for Efficient Satellite FEEL System Design}
{
	In this subsection, we discuss the challenges for efficient FEEL system design. FEEL, as a distributed learning paradigm, typically relies on efficient local model training and frequent model sharing. However, model sharing within integrated space-ground networks exhibits substantial differences compared to ground-only networks. To be specific, short-lived and low-speed GSLs pose two primary challenges:
	
	\begin{enumerate}
		\item \textbf{Long waiting time for feasible GSLs}: 
		To perform FEEL training, the ground PS typically needs to aggregate local models from satellites as well as send the updated global model back to satellites, which requires feasible GSLs for model transmission between the space and the ground. 
		However, given the sporadic nature of ground-space connectivity, such GSLs may not be instantly accessible. Hence, it would cause a long waiting time for feasible GSLs \cite{so2022fedspace,razmi2022scheduling,razmi2022ground} in the global model aggregation and broadcasting phase. 
		\item \textbf{High latency during GSL transmission}: 
		In the global model aggregation step, after waiting for a feasible GSL, model transmission between the space and the ground starts. However, the low data rate of GSLs may cause a high latency for the model transmission. Meanwhile, due to the short duration of GSL, the model transmission may not be able to be completed in one GSL connectivity session and have to experience another long waiting process during which no GSL is feasible. Moreover, both the size of the learning model on each satellite and the total number of satellites in a mega-constellation may be considerably large, resulting in a huge amount of model parameters to be downloaded. This further increases the transmission latency for local model downloading via GSL. 
	\end{enumerate}
}
In view of the challenges delineated above, in the following, we shall propose the \fedmega algorithm.

%
%

\section{Topology-Aware Efficient satellite FEEL Architecture}\label{sec:flsys_design}

In this section, we first present general guidelines of architecture design based on the system features, followed by proposing the \fedmega algorithm which exploits the characteristics of the network topology.

\subsection{Design Principles}
As presented in Section \ref{sec: system model}, efficiently deploying satellite FEEL faces two critical challenges, i.e., long waiting time for feasible GSLs and high GSL transmission latency. A well-designed algorithm should address these two challenges by mitigating the long waiting time and minimizing the transmission latency. 
In light of this, we shall propose two general principles for satellite FEEL system design:
\begin{enumerate}
    \item \textbf{Minimizing GSL utilization}: Recognizing that any utilization of GSL might incur waiting time, the proposed FEEL framework should minimize the frequency of GSL utilization. This approach is intended to alleviate the challenge of excessive waiting time.
    \item \textbf{Reducing GSL transmission load}: To minimize the latency incurred by transmitting models via GSL, the volume of data sent over each GSL should be reduced. This objective can be realized by distributing the total transmission workload across multiple GSLs and enabling collaboration among them. 
\end{enumerate}
In the following, we shall present the \fedmega algorithm, which has been developed based on the above design guidelines. 

\begin{figure}[tbp]
	\centering
	\includegraphics[width=1\linewidth]{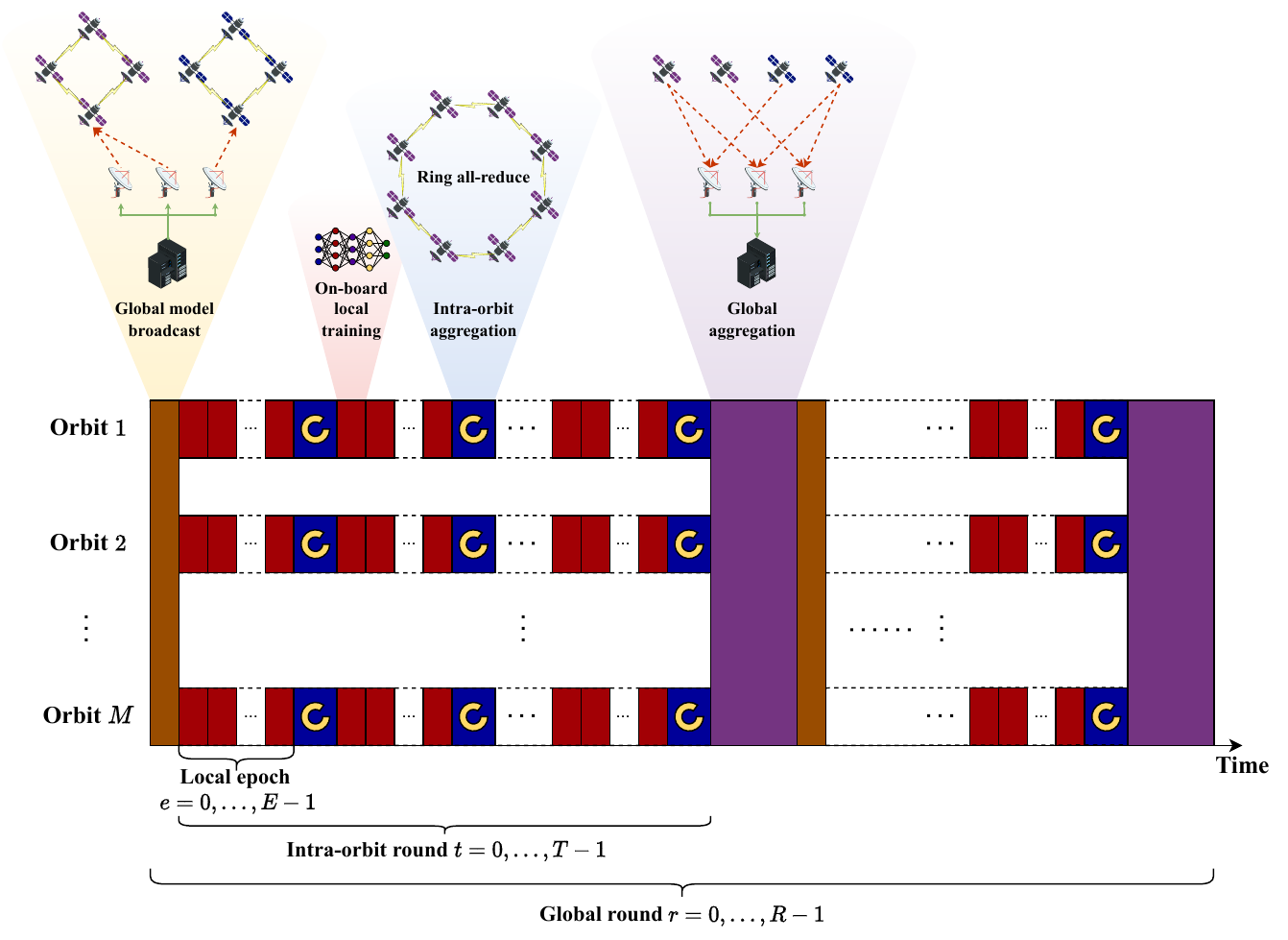}
	\caption{Workflow of \fedmega.}
	\label{fig:fedmega}
	\vspace{-0.5cm}
\end{figure}

\subsection{Proposed \fedmega Framework}
{
	\fedmega is specifically tailored to fit the unique network topology and link characteristics of the LEO network. 
The general workflow of \fedmega is presented in Fig. \ref{fig:fedmega}. Specifically, in each global round, the algorithm procedures are divided into three phases: on-board local training, intra-orbit aggregation, and global aggregation and broadcasting, as elaborated below.
}
\subsubsection{On-Board Local Training}
In this phase, each satellite $\sat_{m,k}$ undergoes local training by performing $E$ steps of mini-batch SGD on a randomly sampled subset of its local dataset, represented by $\xi_{m,k}^{r, t, e} \subseteq D_{m,k}$. At each step $e = 0,1,\dots,E-1$, the local model update is computed as follows:
\begin{equation}
	\bs{z}_{m,k}^{r, t, e+1} = \bs{z}_{m,k}^{r, t, e} - \eta^{r, t, e} \nabla F_{m,k}^{r,t,e}(\bs{z}_{m,k}^{r, t, e}).
\end{equation}
Here, $\eta^{r, t, e}$ represents the stepsize, and $\nabla F_{m,k}^{r,t,e}(\bs{z}_{m,k}^{r, t, e})$ denotes the stochastic gradient of the local loss function $F_{m,k}(\cdot)$ at point $\bs{z}_{m,k}^{r, t, e}$ with respect to the mini-batch $\xi_{m,k}^{r, t, e}$.

\subsubsection{Intra-Orbit Aggregation}
\begin{figure*}[tbp]
	\centering
	\includegraphics[width=1.0\linewidth]{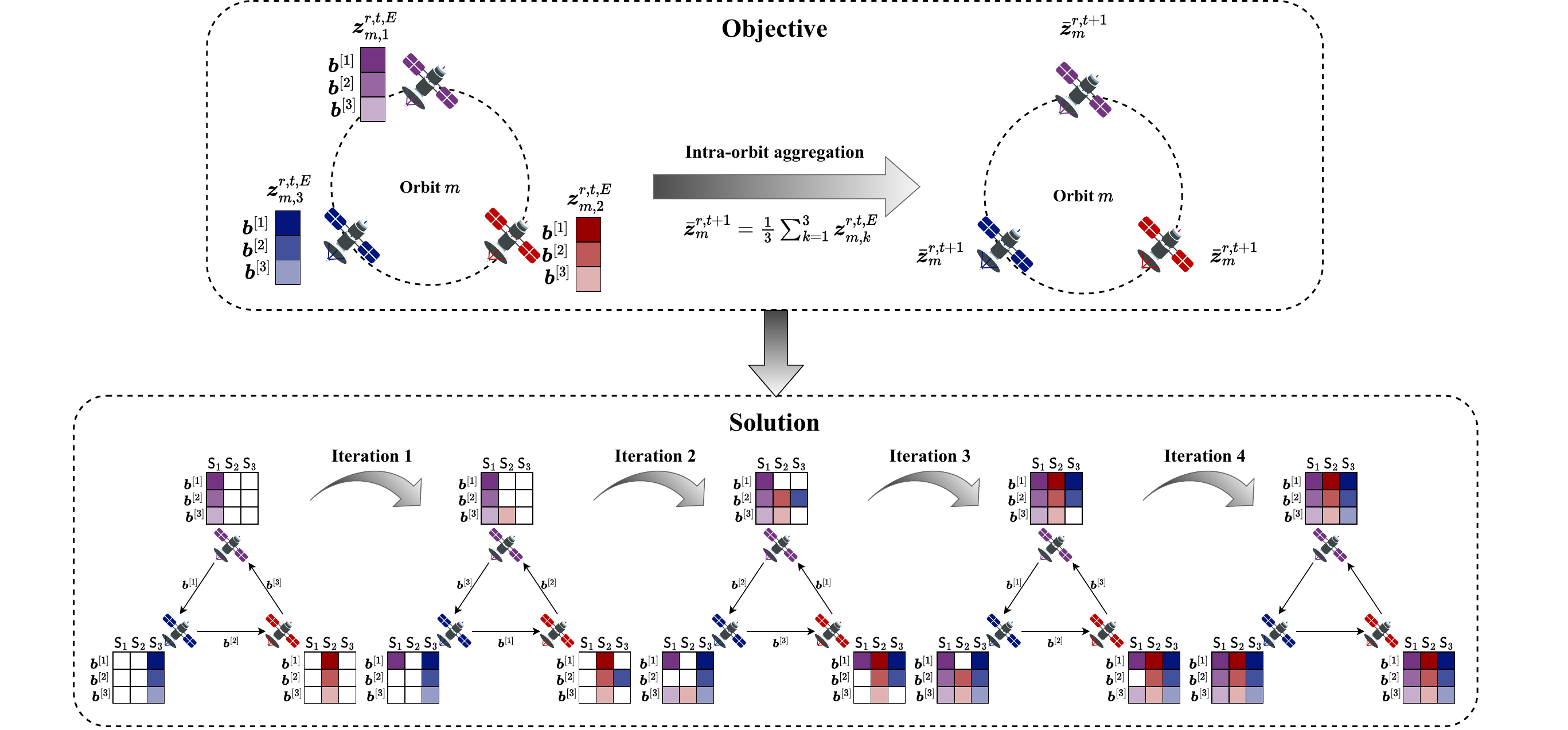}
	\caption{Basic principle of intra-orbit aggregation via ring all-reduce. Consider the orbit $m$ containing $K_m=3$ satellites, i.e., $\sat_{m,1}, \sat_{m,2},\sat_{m,3}$. For notation ease, we omit subscript $m$ and denote them as $\sat_{1}, \sat_{2},\sat_{3}$ in the figure. Each satellite splits its local model into $3$ equal-sized pieces, i.e., $ \bs{b}^{[1]}, \bs{b}^{[2]}, \bs{b}^{[3]}$. In each iteration, all satellites perform simultaneously, and each satellite only transmits one of the 3 pieces to the neighboring satellite, thus requiring only $1/3$ of the time compared to transmitting the whole model. There are a total of $4$ iterations. Note that, for ease of understanding, the workflow in this figure does not utilize the full-duplex feature of laser ISL. Once the full-duplex feature is utilized, as stated in Algorithm \ref{alg:ring-all-reduce}, each satellite splits its local model into 6 pieces: 3 for clockwise transmission and 3 for anti-clockwise transmission, thereby further saving $1/2$ of the time. Each satellite is actually sending and receiving six blocks simultaneously.}
	\label{fig:ringallreduce}
	\vspace{-0.5cm}
\end{figure*}

After performing $E$ rounds of on-board local training, our objective is to attain a periodic averaged intra-orbit global model with the intention of enhancing the performance and generalization abilities of a global model. Specifically, each satellite should attain the intra-orbit model average $\bar{\bm{z}}_{m}^{r,t+1}=\frac{\sum_{k^{\prime} = 1}^{{K}_m}w_{m,k^{\prime}}\bs{z}_{m,k^{\prime}}^{r, t, E}}{\sum_{k^{\prime}=1}^{{K}_m}w_{m,k^{\prime}}} $ as the new start point of local training. 
{
	\SetNlSty{textbf}{}{:}
	\IncMargin{1em}
	\begin{algorithm}[t]
		\setstretch{0.95}
		\textbf{Parameters:} $r$, $t$, $m$.\\
			\ForEach{iteration $i = 0,\cdots,2K_m-3$}
			{					
				\ForPar{satellite $\mathsf{S}_{m,k}\in\mathcal{K}_m$}
				{
					With full-duplex laser ISL, simultaneously do:\\
					$\;\;\;\;$
					Transmit $\bm{b}_{m,k}^{+,[(k-i)\mod K_m]}$ to $\mathsf{S}_{m,(k+1)\mod K_m}$;\\
					$\;\;\;\;$
					Transmit $\bm{b}_{m,k}^{-,[(k+i)\mod K_m]}$ to $\mathsf{S}_{m,(k-1)\mod K_m}$;\\
					\uIf{$i < K_m - 1$}
					{
						$\bm{b}_{m,k}^{+,[(k-1-i)\mod K_m]} \leftarrow \bm{b}_{m,k}^{+,[(k-1-i)\mod K_m]} + \bm{b}_{m,(k-1)\mod K_m}^{+,[(k-1-i)\mod K_m]}$;\\
						$\bm{b}_{m,k}^{-,[(k+1+i)\mod K_m]} \leftarrow \bm{b}_{m,k}^{-,[(k+1+i)\mod K_m]} + \bm{b}_{m,(k+1)\mod K_m}^{-,[(k+1+i)\mod K_m]}$;\\
					}
					\ElseIf{$i \geq K_m - 1$}
					{
						$\bm{b}_{m,k}^{+,[(k-1-i)\mod K_m]} \leftarrow  \bm{b}_{m,(k-1)\mod K_m}^{+,[(k-1-i)\mod K_m]}$;\\
						$\bm{b}_{m,k}^{-,[(k+1+i)\mod K_m]} \leftarrow  \bm{b}_{m,(k+1)\mod K_m}^{-,[(k+1+i)\mod K_m]}$;\\
					}	
				}
			}
			\ForPar{satellite $\mathsf{S}_{m,k}\in\mathcal{K}_m$}
			{
				$\bm{z}_{m,k}^{r,t+1,0} \leftarrow \mathrm{Concatenate}\left (\bm{b}_{m,k}^{+,[0]},\bm{b}_{m,k}^{-,[0]},\bm{b}_{m,k}^{+,[1]}, \dots, \bm{b}_{m,k}^{-,[K_m-1]}\right )$;
			}
		\caption{Ring All-Reduce for Intra-Orbit Aggregation with Full-Duplex Laser ISL.}
		\label{alg:ring-all-reduce}
	\end{algorithm}
}	
However, in the LEO network, the frequent transmission of all local models to the ground PS for global model aggregation introduces significant waiting time and transmission latency due to the usage of GSLs. Nevertheless, as detailed in the system model in Section \ref{sec: system model}, laser ISLs provide stable and rapid connections among satellites within the same orbit, resulting in a ring topology. This advantageous characteristic inspires us to explore the potential of utilizing laser ISLs for efficient intra-orbit aggregation within the orbit, consequently attaining a certain level of aggregation effectiveness. This intra-orbit aggregation leverages the stable, high-speed laser ISLs within the orbit, eliminating the significant transmission latency and conserving energy. Following several rounds of intra-orbit aggregation, a global aggregation is performed. This design significantly diminishes the waiting time for GSLs during the aggregation process and accelerates model training by reducing the utilization frequency of GSLs.

\begin{figure}[t]
	\centering
	\includegraphics[width=0.80\linewidth]{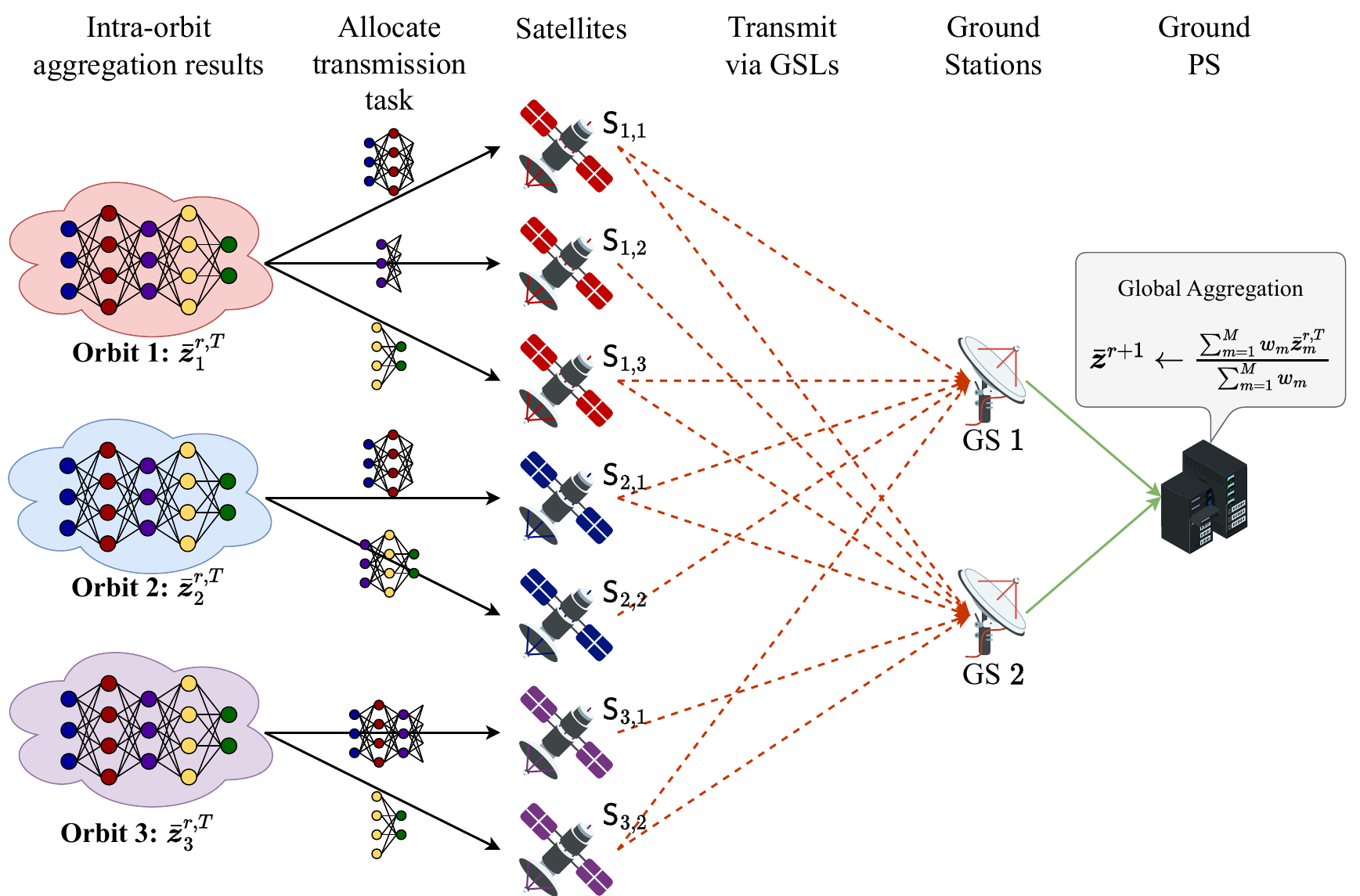}
	\caption{Main idea of distributed model transmission: In each orbit $m$, since all satellites have the same intra-orbit aggregation result $\bar{\bs{z}}_{m}^{r,T}$, the intra-orbit aggregation result can be partitioned into various parts and each satellite transmits one of these parts to the ground to accomplish the model downloading.}
	\label{fig:distributed-transmission-formal-real}
	\vspace{-0.5cm}
\end{figure}

\begin{figure}[t]
	\centering
	\includegraphics[width=0.9\linewidth]{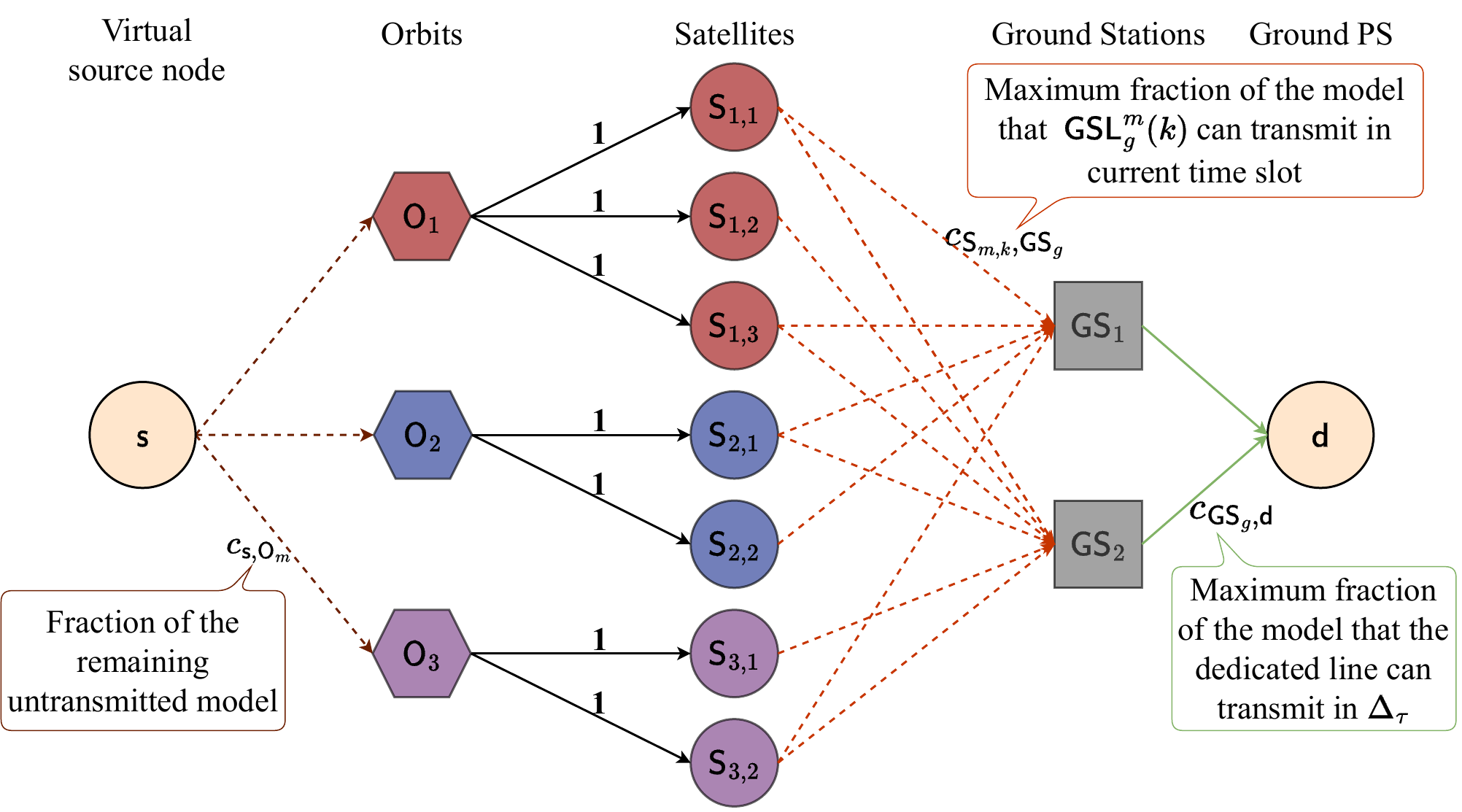}
	\caption{The formulated network flow problem corresponding to the distributed transmission: In each time slot, we aim to maximize the fraction of received model at the ground PS. $\mathsf{s}$ is the source node $\mathsf{d}$ is the destination node in the network flow problem. The capacity between orbit nodes and satellite nodes are all set to $1$ because each satellite can transmit at most $100\%$ fraction of the intra-orbit aggregated model. The capacity of other edges are set and explained as shown in the figure.}
	\label{fig:distributed-transmission-formal}
	\vspace{-0.5cm}
\end{figure}
{
	Moreover, in the \fedmega algorithm, the ring topology of satellites in the same orbit is leveraged, utilizing the ring all-reduce method from high-performance computing \cite{sergeev2018horovod} for model aggregation. The operations and principle of the ring all-reduce method are illustrated in Fig. \ref{fig:ringallreduce}, which depicts a half-duplex scenario for ease of understanding. The procedure of ring all-reduce aided intra-orbit aggregation is provided in Algorithm \ref{alg:ring-all-reduce}, where the full-duplex feature of laser ISL is utilized to double the aggregation speed compared to conventional ring all-reduced methods. 
	
}
Specifically, without loss of generality, we assume that $d$ is divisible by $2K_m$. Then, each satellite $\sat_{m, k}$ partitions its local model $\bm{z}_{m,k}^{r,t,E}$ into $2K_m$ chunks, i.e., $\bm{b}_{m,k}^{+,[0]},\bm{b}_{m,k}^{-,[0]},\bm{b}_{m,k}^{+,[1]}, \dots, \bm{b}_{m,k}^{-,[K_m-1]}$, where
\begin{equation}
	\left\{
	\begin{aligned}
		&\bm{b}_{m,k}^{+,[u]} = \frac{w_{m,k}}{\sum_{k^{\prime}=1}^{K_m}w_{m,k^{\prime}}}\bm{z}_{m,k}^{r,t,E}\left[\frac{ud}{K_m}+1:\frac{(u+1/2)d}{K_m}\right]\\
		&\bm{b}_{m,k}^{-,[u]} = \frac{w_{m,k}}{\sum_{k^{\prime}=1}^{K_m}w_{m,k^{\prime}}}\bm{z}_{m,k}^{r,t,E}\left[\frac{(u+1/2)d}{K_m}+1:\frac{(u+1)d}{K_m}\right],
	\end{aligned}
	\right.
\end{equation}
and $u=0,1,\dots,K_m-1$.
Then, there are totally $2K_m-2$ iterations. In each iteration $i$, each satellite only needs to transmit one chunk per laser ISL, and hence only needs $ \frac{1}{2K_m}$ time compared to transmit the whole model. The total time cost $\frac{2K_m-2}{2K_m}$ is upper bounded by a constant (e.g., 1) irrelevant to $K_m$, which implies that the intra-orbit model aggregation step is scalable with respect to the constellation scale. Consequently, orbits with different numbers of satellites are expected to incur almost the same time cost for intra-orbit aggregation, which enhances the synchronization of the system. As a result, each satellite finally acquires the averaged model from all satellites within its orbit, allowing them to proceed to the next intra-orbit round's local training or global aggregation.
Ideally, our algorithm achieves optimal time complexity when data rates of the ISL are uniform. However, when the inter-satellite link rates vary between different satellites and change over time in real systems, the overall performance is limited by the slowest one.

\subsubsection{Global Aggregation and Broadcasting}


After multiple intra-orbit rounds, the algorithm advances to the global aggregation phase, which involves aggregating all satellites' local models via GSLs. The updated global model is then disseminated to all satellites to initiate a new global round. As previously mentioned, the long transmission latency presents a core challenge during this phase. To this end, a two-step strategy is employed in \fedmega, i.e.,
	\begin{enumerate}

        \item First, given that all satellites within a given orbit share an identical intra-orbit model upon completing the ring all-reduce-assisted intra-orbit aggregation, the ground PS merely needs to communicate with any single satellite in that orbit to obtain the corresponding intra-orbit aggregation outcome. This method avoids the need to retrieve each satellite's local model individually, thereby considerably reducing the data download volume. 
		\item Secondly, as multiple GSLs may be feasible at the same time, a network flow based distributed transmission scheme is proposed. This method involves partitioning the model downloading task into sub-tasks, where multiple feasible GSLs from LEO satellites in one orbit collaboratively execute the model downloading, as illustrated in Fig. \ref{fig:distributed-transmission-formal-real}. 
		Consequently, the load on each GSL is alleviated, leading to a substantial reduction in the overall transmission latency.
	\end{enumerate}


{
	\SetNlSty{textbf}{}{:}
	\IncMargin{1em}
	\begin{algorithm}[t]
 \setstretch{0.95}
		\textbf{Parameters}: $r$, $\Delta_\tau$, $ \theta $.\\
		\textbf{Input}: constructed graph $\mathscr{G}$.\\
		\While{$\exists m \in \mathcal{M} $ s.t. $c_{\mathsf{s}, \mathsf{O}_m} \neq 0$}
		{
			Delete all edges $(\sat_{m, k}, \mathsf{GS}_g), \forall m,k,g$;\\
			Detect all feasible $\gsl_{g}^{m}(k)$ in current time slot;\\
			\For{each feasible $\gsl_{g}^{m}(k)$}
			{
				$r_{\sat_{m, k}, \mathsf{GS}_g} \leftarrow $ data rate in current time slot;\\
				Add an edge $(\sat_{m, k}, \mathsf{GS}_g)$ to $\E$ with capacity $c_{\sat_{m, k}, \mathsf{GS}_g} \leftarrow \frac{r_{\sat_{m, k}, \mathsf{GS}_g}\Delta_\tau}{\theta}$;\\
			}
			Initialize flow $f_{\mathsf{u}, \mathsf{v}} \gets 0$ for all edges $(\mathsf{u}, \mathsf{v}) \in \E$\;
			\While{there exists an augmenting path $p$ from $\mathsf{s}$ to $\mathsf{d}$ in the residual graph $\mathscr{G}_f$}{
				Denote $c_f(\mathsf{u}, \mathsf{v})$ as the residual capacity of edge $(\mathsf{u}, \mathsf{v})$ in $\mathscr{G}_f$;\\
				$c_f(p) \leftarrow \min_{(\mathsf{u}, \mathsf{v}) \in p} c_f(\mathsf{u}, \mathsf{v})$\;
				\For{each edge $(\mathsf{u}, \mathsf{v}) \in p$}{
					$f_{\mathsf{u}, \mathsf{v}} \leftarrow f_{\mathsf{u}, \mathsf{v}} + c_f(p)$\;
					$f_{\mathsf{v}, \mathsf{u}} \leftarrow f_{\mathsf{v}, \mathsf{u}} - c_f(p)$\;
				}
			}
			\For{each feasible $\gsl_{g}^{m}(k)$}
			{
				Transmit $f_{\sat_{m, k}, \mathsf{GS}_g}$ fraction of the model $\bar{\bz}_{m}^{r,T}$ during current time slot;\\
			}
			$c_{\mathsf{s}, \mathsf{O}_m} \leftarrow c_{\mathsf{s}, \mathsf{O}_m} - f_{\mathsf{s}, \mathsf{O}_m}$;\\
		}

		\caption{Distributed and Adaptive Transmission Scheme for Global Aggregation.}
		\label{alg:FedMega:globalaggregation}
	\end{algorithm}
}
Specifically, at the onset of global aggregation, time is partitioned into discrete time slots. Within each individual time slot, the feasibility and transmission rate of all GSLs are invariant. Consequently, the ground PS is able to acquire a graph $\mathscr{G}$, which represents the present network connectivity, as illustrated in Fig. \ref{fig:distributed-transmission-formal}. In this representation, each satellite $\mathsf{S}_{m,k} \in \mathcal{K}$ corresponds to vertex $\mathsf{S}_{m,k}$, each GS $g$ is represented by vertex $\mathsf{GS}_g$, the ground PS is denoted by vertex $\mathsf{d}$, and each edge $(\mathsf{v}_1, \mathsf{v}_2)$ in the right portion signifies a feasible link between $\mathsf{v}_1$ and $\mathsf{v}_2$ within the current time slot. The capacity $c_{\mathsf{v}_1, \mathsf{v}_2}$ of this edge is equivalent to the maximum fraction of a model that the link can transmit in the current time slot, expressed as $c_{\mathsf{v}_1, \mathsf{v}_2} = \frac{r_{\mathsf{v}_1, \mathsf{v}_2}\Delta\tau}{\theta}$, where $r_{\mathsf{v}_1, \mathsf{v}_2}$ denotes the data rate of the link in the current time slot, $\Delta\tau$ represents the duration of the time slot, and $\theta$ indicates the model size.
{
	\SetNlSty{textbf}{}{:}
	\IncMargin{1em}
	\begin{algorithm}[t]
            \setstretch{0.95}
		PS: Initialize $\bar{\bs{z}}^{0}$;\\
		\ForEach{global round $r=0,\cdots,R-1$}
		{
			\ForPar{orbit $m\in\mathcal{M}$}
			{
				\ForPar{satellite $\mathsf{S}_{m,k}\in\mathcal{K}_m$}
				{
					$ \bs{z}_{m,k}^{r,0,0} \leftarrow \bar{\bs{z}}^{r}  $; 
				}
				\ForEach{intra-orbit round $t=0,\cdots,T-1$}
				{
					\ForPar{satellite $\mathsf{S}_{m,k}\in\mathcal{K}_m$}
					{
						\ForEach{local epoch $e=0,\cdots,E-1$}
						{
							$	\bs{z}_{m,k}^{r, t, e+1} \leftarrow \bs{z}_{m,k}^{r, t, e} - \eta^{r, t, e} \nabla F_{m,k}^{r,t,e}(\bs{z}_{m,k}^{r, t, e})$;
						}
					}
					$ \bar{\bs{z}}_{m}^{r,t+1} \leftarrow   \frac{\sum_{k^{\prime} = 1}^{{K}_m}w_{m,k^{\prime}}\bs{z}_{m,k^{\prime}}^{r, t, E}}{\sum_{k^{\prime}=1}^{{K}_m}w_{m,k^{\prime}}} $;\\
					\ForPar{satellite $\mathsf{S}_{m,k}\in\mathcal{K}_m$}
					{
						$ \bs{z}_{m,k}^{r,t+1,0} \leftarrow \bar{\bs{z}}_{m}^{r,t+1}  $ using Algorithm \ref{alg:ring-all-reduce};
					}
				}
			}
			$\bar{\bs{z}}^{r+1} \leftarrow \frac{\sum_{m=1}^{M}w_{m}\bar{\bs{z}}_{m}^{r, T}}{\sum_{m=1}^{M}w_{m}}$ using Algorithm \ref{alg:FedMega:globalaggregation};\\
		}
		\KwOut{$\bar{\bs{z}}^{R}$.}
		\caption{The Proposed \fedmega Algorithm.}
		\label{alg:FedMega}
	\end{algorithm}
	
}

To expedite the downloading process, we maximize the data volume received by ground PS in each time slot, which is formulated as a network flow problem. The construction of this network flow problem can be described as follows, with reference to Fig. \ref{fig:distributed-transmission-formal}. A vertex $\mathsf{s}$ represents a virtual source node, while each orbit $m\in\mathcal{M}$ is denoted by a vertex $\mathsf{O}_m$.
Every vertex $\mathsf{O}_m$ is connected to all vertices corresponding to satellites in orbit $m$. This signifies that each satellite can only transmit the aggregated intra-orbit model of the orbit to which it belongs. All edges between satellite vertices and orbit vertices have a constant capacity of one, implying that the proportion of the intra-orbit model allocated to each satellite to transmit can not exceed 100\%.
Additionally, the source vertex $\mathsf{s}$ is connected to all orbit vertices, with the edge capacity $c_{\mathsf{s}, \mathsf{O}_m}$ representing the fraction of orbit $m$'s intra-orbit model that remains untransmitted as of the current time slot.
Hence, at the beginning of each time slot, the ground PS can schedule the transmission task of each GSL by solving this network flow problem, until all orbits' intra-orbit models have been transmitted. The distributed and adaptive transmission scheme is summarized in Algorithm \ref{alg:FedMega:globalaggregation}. 
After aggregation of the intra-orbit models from all orbits, we obtain an updated global model. The ground PS transmits this updated global model back to all satellites with a similar distributed and adaptive uploading scheme.

In summary, the proposed satellite FEEL framework includes on-board local training, intra-orbit aggregation, and global aggregation and broadcasting, as summarized in Algorithm \ref{alg:FedMega}.

\subsection{Advantages of \fedmega Framework}
{
	The proposed \fedmega aligns with the core design guidelines, i.e., GSL utilization minimization and transmission load reduction, thereby speeding up the FEEL training process. The implementation of multiple intra-orbit rounds can reduce the gradient estimation variance within one orbit within very short time. This effectively mitigates the frequency of GSL utilization, consequently diminishing the waiting time required for feasible GSLs. 
	In terms of intra-orbit model aggregation, the ring all-reduce-based scheme significantly reduces the communication overhead, i.e., time and energy cost.
	Furthermore, the adoption of a distributed and adaptive model transmission scheme alleviates the transmission load burden on each GSL, thereby accelerating the downloading process.
	%
	%
	
}

	\section{Convergence Analysis for \fedmega}\label{sec: convergence}
In this section, we present the convergence analysis of the \fedmega algorithm. 
To ensure that our analysis is relevant to practical scenarios (e.g., deep neural networks), we focus on the situations where non-convex loss functions and non-IID data are present. In order to make convergence analysis tractable, we make the following assumptions.
\begin{ass}\label{ass: lower}
	The objective function $F$ is lower bounded by a scalar $F^{\inf}$, i.e., $F(\bm{x})\geq F^{\inf}>-\infty$.
\end{ass}

\begin{ass}[$L$-Smoothness] \label{ass: smooth}
	The loss function of each satellite $\mathsf{S}_{m, k}$, each orbit $m$, and all orbits, i.e., $F_{m,k}(\cdot)$, $ F_{m}(\cdot) $, and $ F(\cdot) $ are all $L$-smooth. For all $\bm{x}$ and $\bm{y}$, we have
	\begin{equation*}
		\|\nabla F_{m,k}(\bm{x})-\nabla F_{m,k}(\bm{y})\|_2 \leq L \|\bm{x} - \bm{y}\|,
	\end{equation*}
	and 
	$$F(\bm{y})\leq F(\bm{x}) + \left\langle \nabla F(\bm{x}),\bm{y-x}\right\rangle+\frac{L}{2}\left\|\bm{y-x}\right\|_2^2.$$
\end{ass}


\begin{ass}[Unbiased Gradient and Bounded Variance] \label{ass: gradient variance}
	Each satellite can query an unbiased stochastic gradient with bounded variance, i.e., $$\mathbb{E}[\nabla F_{m,k}^{r,t,e}(\bm{x}) ] = \nabla F_{m,k}(\bm{x}),$$$$\mathbb{E} \left[ \| \nabla F_{m,k}^{r,t,e}(\bm{x}) - \nabla F_{m,k}(\bm{x}) \|_2^2 \right] \leq \sigma^2 .$$

\end{ass}

\begin{ass}[Intra-orbit Dissimilarity]\label{ass: intra-orbit dissimilarity}
	For each orbit $m\in \mathcal{M}$, there exists a constant $\delta_m \geq 0$ such that
	\begin{align}\label{eq: intra-orbit error}
		\frac{1}{K_0}\sum_{k=1}^{K_0}\|\nabla   F_{m,k}(\bm x) - \nabla   F_{m}(\bm x)\|_2^2\leq \delta_m^2.
	\end{align}
\end{ass}

\begin{ass}[Inter-orbit Dissimilarity]\label{ass: inter-orbit dissimilarity}
	There exists constants $\alpha \geq 1$ and $\beta\geq 0$ such that
	\begin{align}\label{eq: common bound}
		\frac{1}{M}\sum_{m=1}^{M}\|\nabla   F_m(\bm x)\|_2^2\leq \beta+\alpha\|\nabla F(\bm x)\|_2^2.
	\end{align}
\end{ass}

Assumptions \ref{ass: lower}-\ref{ass: gradient variance} are common in stochastic optimization \cite{bottou2018optimization}.
Assumption \ref{ass: intra-orbit dissimilarity} is adopted to characterize the gradient dissimilarity between each satellite's loss function and the corresponding orbit's loss function \cite{guo2022hybrid}.
Assumption \ref{ass: inter-orbit dissimilarity}, known as the bounded gradient dissimilarity \cite{karimireddy2020scaffold} measures the non-IID degree of data on different devices.

We first provide the convergence results of \fedmega with respect to the squared norm of the gradient $\|\nabla F(\bar{\bm z}^{r,t})\|_2^2$ in Theorem \ref{thm: theorem 1}, which is an effective metric in non-convex optimization \cite{bottou2018optimization}. 

\begin{thm}[Convergence of \fedmega Under Non-convexity]\label{thm: theorem 1}
	Let Assumptions \ref{ass: smooth}, \ref{ass: gradient variance}, \ref{ass: intra-orbit dissimilarity}, and \ref{ass: inter-orbit dissimilarity} hold. Choosing a learning rate as 
	\begin{align}\label{eqn:stepsize}
		0<\eta\leq \begin{cases}
			\frac{1}{2LE}&E = 1,T=1,\\
			\frac{1}{8L\sqrt{5/3\alpha E(E-1)}} & E\geq 2, T = 1,\\
			\frac{1}{8EL\sqrt{3\alpha T(T-1)}} & E\geq1, T\geq2,
		\end{cases}
	\end{align}
	the expected norm of global gradients after $R$ communication rounds is upper bounded by
	\begin{align*}
		&\frac{1}{RT}\sum\limits_{r = 0}^{R-1}\sum\limits_{t = 0}^{T-1} \mathbb{E} \left[\left \|\nabla F(\bar{\bm z}^{r,t})\right \|_2^2\right]\\
		\leq & \frac{4\left[F(\bm{z}^0) - F^{\inf}\right] }{\eta ETR} 
		+ \frac{160 }{3K}\eta^3L^3E(M-1)(T-1)\sigma^2\\
		+& 16\eta^2L^2(E-1)(\sigma^2+6\bar{\delta}^2)\\
		+& \frac{32}{3}\eta^2L^2E\beta\left[5ET(T-1)+9(E-1)\right]
		+  \frac{ 4}{K}\eta L \sigma^2.
	\end{align*}
	where $\bar{\delta}^2 = \frac{1}{M}\sum_{m=1}^M \delta_m^2$.
\end{thm}
 \begin{proof}
 	Please refer to the Appendix.
 \end{proof}
Based on Theorem \ref{thm: theorem 1}, we obtain the following key observations.
The first term of the upper bound demonstrates that the optimality gap depends on the initial point.
The second term is caused by bounded variance, multiple orbits, and multiple intra-orbit aggregations.
The third term reflects the impact of multiple local updates.
The fourth term shows the coupling of multiple local updates, multiple intra-orbit aggregations, and the non-IID data distribution.
The last term is caused by the gradient variance.
\begin{Rem}\label{rem: algorithm reduction}
	By setting $T = 1$, \fedmega reduces to the conventional \fedavg with multiple local updates \cite{FedAvg}. Furthermore, by setting $E = 1,T = 1$, \fedmega recovers the results of the fully synchronized SGD as in \cite{bottou2018optimization}. By setting $E = 1,T = \tau$, \fedmega also reduces to the fully connected case (i.e., the mixing matrix of D2D communication links is an all-ones matrix) of the \hlsgd \cite{guo2022hybrid}.
\end{Rem}

By designing the learning rate, we can derive the following corollary with $E, T, R \gg 1$.

\begin{cor}[Error Bound of \fedmega under Non-convexity]\label{cor: fedmega}
	When the learning rate is a constant satisfying $ \eta = \frac{1}{L}\sqrt{\frac{K}{ETR}} $,
	the minimal gradient norm of the global objective function of \fedmega is bounded as follows
	\begin{align}\nonumber
		&\min_{t\in[T],r\in[R]} \mathbb{E} \left[\left \|\nabla F(\bar{\bm z}^{r,t})\right \|_2^2\right]
		\leq\mathcal{O}\left( \frac{1+\sigma^2}{\sqrt{KETR}}\right)\\
		+& \mathcal{O}\left( \frac{M-1}{R}\sqrt{\frac{K}{ETR}}\right)
		+ \mathcal{O}\left(\frac{K [\sigma^2+\bar{\delta}^2+E(T^2+1)]}{TR}\right).
	\end{align}
	The right hand side of the inequality is dominated by $\mathcal{O}\left( \frac{1+\sigma^2}{\sqrt{KETR}}\right)$.
\end{cor}

To measure the convergence performance with respect to communication efficiency and the parameters that determining the architecture design of \fedmega (i.e., $E, T$), we have the following remark for discussion.

\begin{Rem}\label{rem: convergence rate}
	\fedmega achieves a linear speed-up in terms of the number of LEO satellites (i.e., $ K $), the number of local updates (i.e., $E$), and the number of intra-orbit aggregations (i.e., $ T $).
	We say that solution $\bm{z}^t$ is an $\epsilon$-stationary point if $\left \|\nabla f(\bm{z}^{t})\right \|_2^2\leq\epsilon$. According to Corollary \ref{cor: fedmega}, the learning rate $ \eta = \frac{1}{L}\sqrt{\frac{K}{ETR}} $ decreases with the number of local updates and the number of intra-orbit aggregations. To achieve an $\epsilon$-stationary solution, the \fedmega algorithm takes $ETR = \frac{(1+\sigma^2)^2}{K\epsilon^2}$ total iterations. If we increase the product of the number of local updates and intra-orbit aggregations, then we can obtain a solution with higher accuracy and smaller learning rate occupying the same number of communication rounds (i.e., $ R $). In addition, if $ETR$ is fixed, we can increase the product $ET$ to reach the same level of $\epsilon$-stationary solution, which results in reducing the number of communication rounds to enhance the communication efficiency. However, due to data heterogeneity, $ET$ can not be infinite large, and the optimal choice of $ET$ is $ \mathcal{O}(\frac{R}{K}) $.
\end{Rem}

{
	To conclude, we minimize the frequency of GSL utilization by implementing multiple rounds of on-board local training and intra-orbit aggregations, which accelerates the convergence of \fedmega as analyzed in Remark \ref{rem: convergence rate}. 
	To compare the convergence rate of \fedmega with existing studies, we have the following remark.}

\begin{Rem}\label{rem: algorithm comparison}
	As a matter of fact, \fedmega achieves a hierarchical distributed SGD convergence rate as in \cite{wang2022demystifying,liu2023hierarchical}, which is $\sqrt{T}$ times faster than \fedavg and its implementation on LEO satellites (i.e., \fedisl) \cite{razmi2022board} in terms of total iterations. \fedmega is also $\sqrt{E}$ times faster than the upper bound of the \hlsgd \cite{guo2022hybrid}. Besides, to support the three-layer client-edge-cloud hierarchical FEEL on the ground, more intra-orbit aggregations cause frequently device-to-server communications, which is definitely not communication-efficient. However, a hierarchical distributed SGD convergence rate can be realized on satellite FEEL over the LEO mega-constellations with only two-layer orbit-ground architecture by fully exploiting the parallel computing property of ring all-reduce for intra-orbit aggregations. Thus \fedmega is both communication-efficient and structure-efficient.
	
\end{Rem}

\begin{figure}[!]
	\begin{subfigure}{0.24\textwidth}
		\centering
		\includegraphics[width=0.9\linewidth]{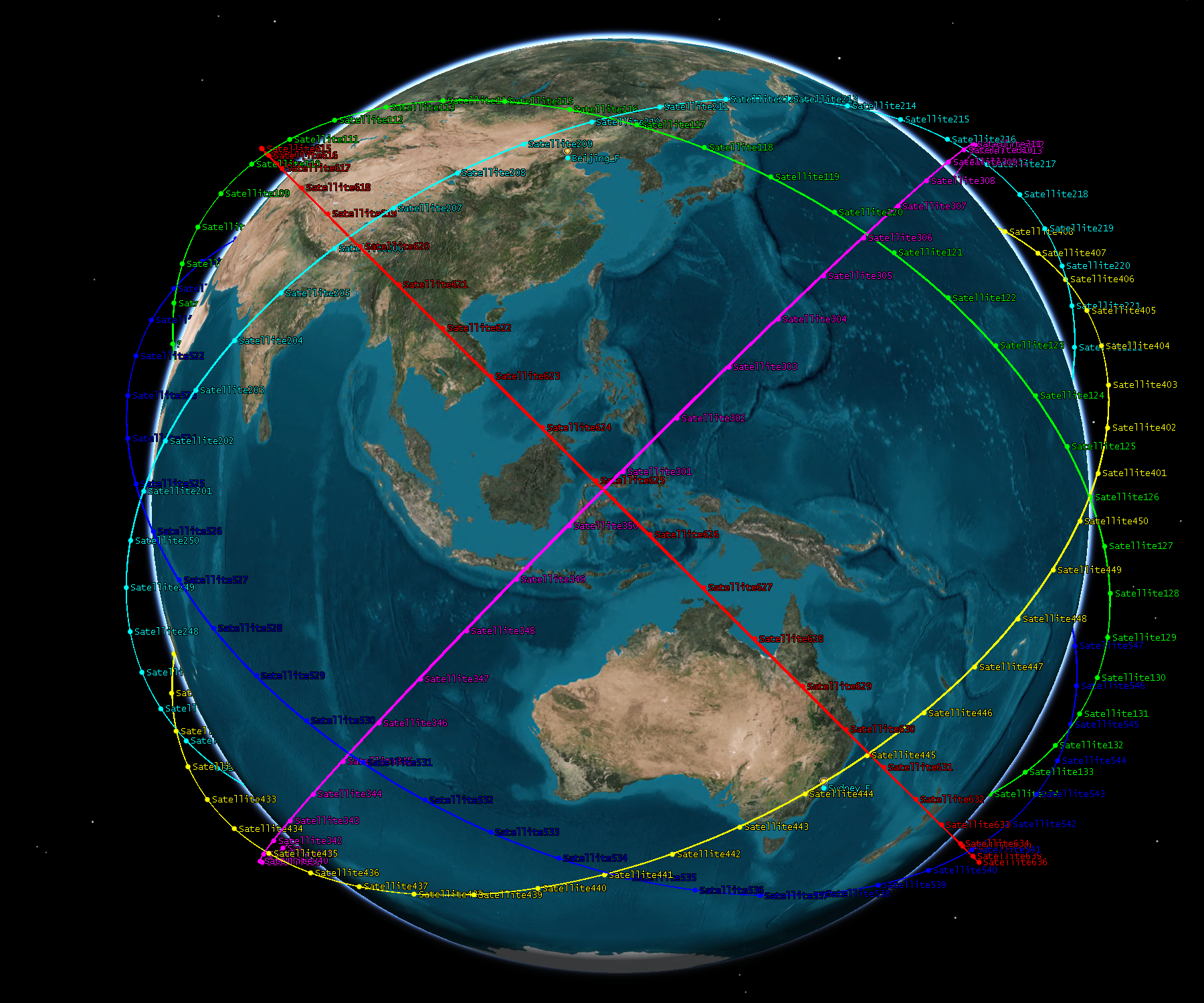}
		\caption{The $300/6/1$ Walker Delta constellation.}
		\label{fig:c1}
	\end{subfigure}
	\hfill
	\begin{subfigure}{0.24\textwidth}
		\centering
		\includegraphics[width=0.9\linewidth]{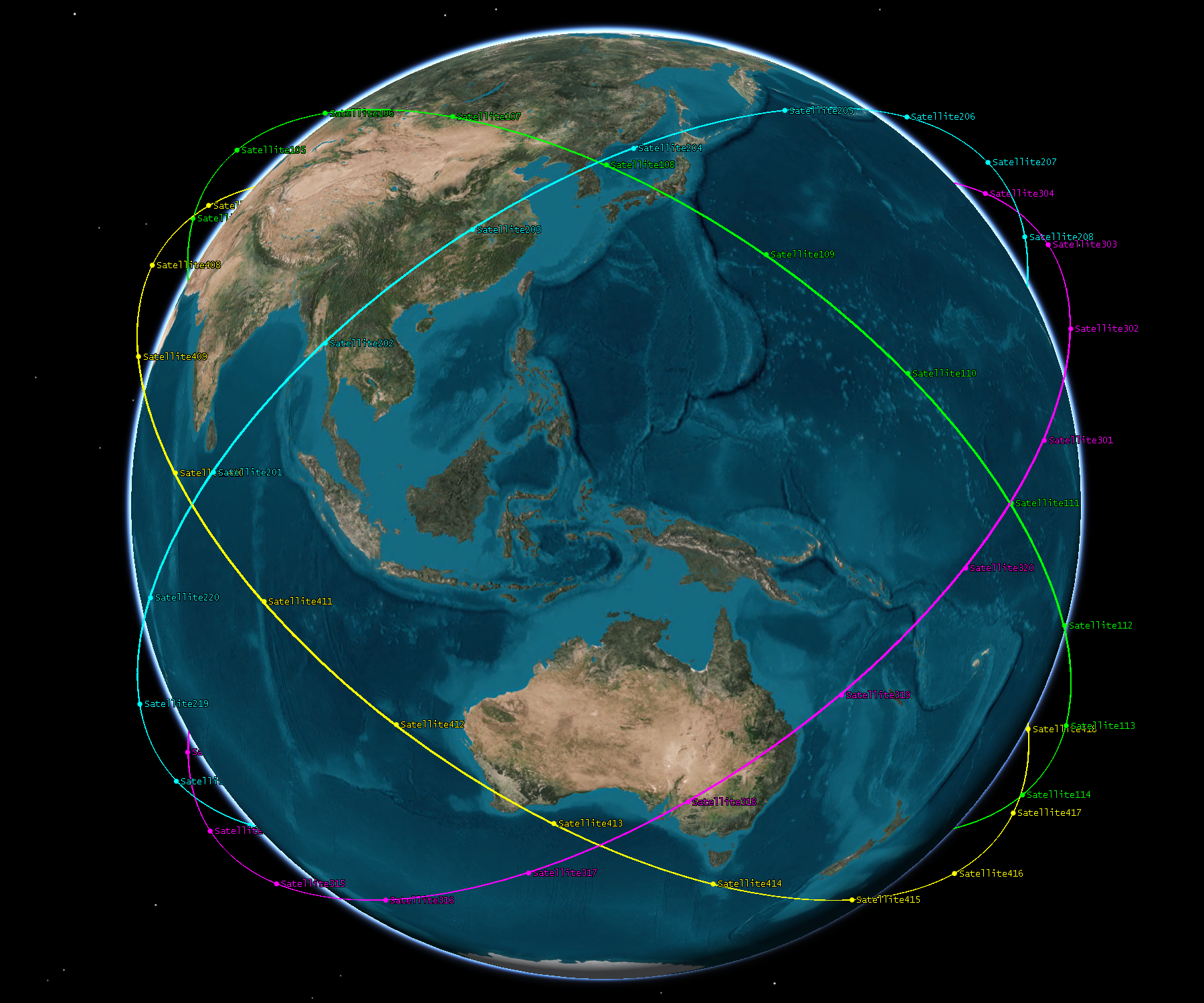}
		\caption{The $80/4/1$ Walker Delta constellation.}
		\label{fig:c2}
	\end{subfigure}
	\caption{Two different satellite constellation network settings.}
	\label{fig:constellation}
	\vspace{-0.5cm}
\end{figure}

\section{Simulations}\label{sec: simulations}
We evaluate the performance of the proposed satellite FEEL algorithm, i.e., \fedmega, for different learning tasks and in different network settings in this section.
\subsection{Experiment Setup}\label{sec: setup}
\subsubsection{LEO Mega-Constellation Network Configurations}
To verify the performance of the satellite FEEL algorithms, we consider two different network settings, as visualized in Fig. \ref{fig:constellation}. The first one is a $300/6/1$ Walker Delta constellation system and the second one is a $80/4/1$ Walker Delta constellation system. In the $300/6/1$ Walker Delta setting, there are $M=6$ orbits and $K_0 = 50$ satellites evenly distributed in each orbit. In $80/4/1$ Walker Delta setting, there are $M=4$ orbits and $K_0 = 20$ satellites evenly distributed in each orbit. 

For both systems, all the satellites are at a $ 500 $ km altitude and there are $ G = 6 $ GSs located at Beijing ($39.9289$°N, $116.388$°E), Berlin ($52.5167$°N, $13.4$°E), Cape Town ($-33.9167$°N, $18.4167$°E), Rio De Janeiro ($-22.9$°N, $-43.2333$°E), Syndey ($-33.8833$°N, $151.217$°E), Toronto ($43.6667$°N, $-79.4167$°E), respectively. 

%
\begin{figure}[!]
	\centering
	\includegraphics[width=1\linewidth]{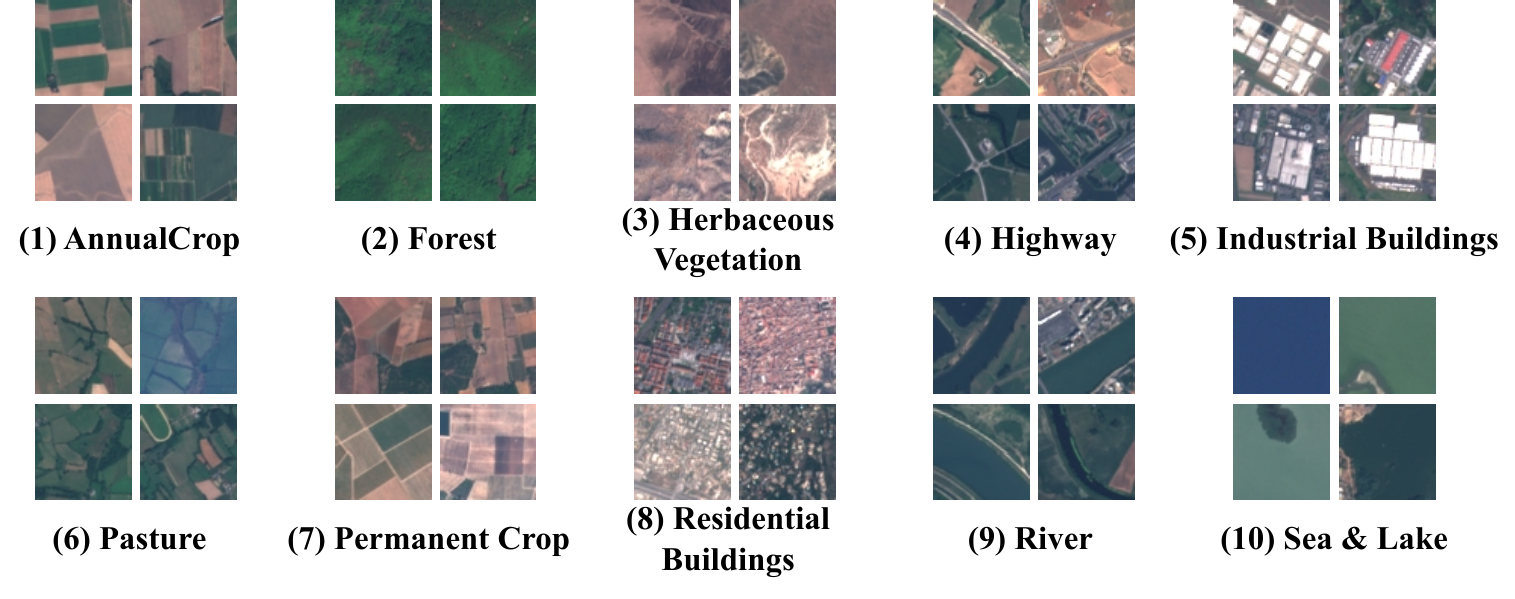}
	\caption{The EuroSAT dataset.}
	\label{fig:eurosat}
	\vspace{-0.2cm}
\end{figure}

\begin{figure}[!]
	\centering
	\includegraphics[width=0.95\linewidth]{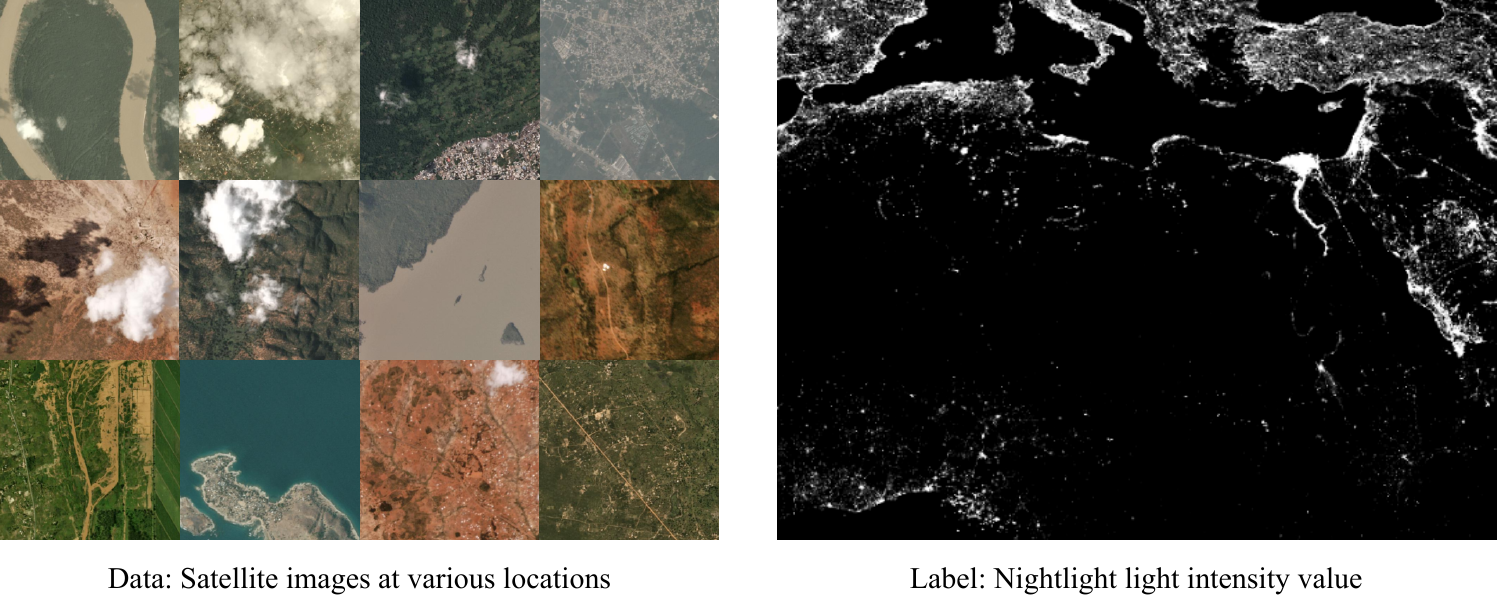}
	\caption{The dataset of nighttime light intensity prediction task.}
	\label{fig:real}
	\vspace{-0.5cm}
\end{figure}

\subsubsection{Learning Tasks}
The following three learning tasks are used to evaluate the performance of the FEEL algorithms.
\begin{itemize}
	\item \textbf{Synthetic data classification\cite{FedProx}}. 
	For this task, we train a deep neural network (DNN) on $\operatorname{Synthetic}(0.5,0.5)$ dataset generated in \cite[Section 5.1]{FedProx}, which contains 10 labels of data. The DNN is a two-layer feedforward neural network with 60 input neurons, 20 hidden neurons, and 10 output neurons. The local data generated by each satellite is assumed to be non-independent and identically distributed (non-IID) in each orbit.
	In this experiment, each satellite contains $\operatorname{Unif}(50, 450)$ data samples. The batch-size is set to be 25. The learning rate is set to be 0.01. The number of communication round $R$, the number of intra-orbit rounds $T$, and the number of local update $E$ are set to be 600, 10, and 5, respectively.
	\item \textbf{Land cover classification on EuroSAT \cite{helber2019eurosat}}.
	In this task, we train a ResNet-18 neural network model on the Sentinel-2 satellite RGB images to classify the land cover and land use. 
 In fact, satellites are capable of retrieving land cover labels for remote sensing images by matching the geographical location information of captured images with the geographic information database on the ground. This enables the acquisition of labels for remote sensing images to be used for training purposes.
    Specifically, the dataset contains 27,000 images with 10 classes of land cover illustrated in Fig. \ref{fig:eurosat}. The local data of each satellite is assumed to be IID sized $\operatorname{Unif}(300, 375)$. The batch-size is set to be $16$. The learning rate is set to be $1\times 10^{-4}$. The number of intra-orbit rounds $T$, and the number of local update $E$ are set to be $5$, and $20$, respectively.
	\item \textbf{Nighttime light intensity prediction.} 
	The aim of this task is to train a neural network model on daytime satellite images to predict the nighttime light intensity level of the corresponding area \cite{jean2016combining}. As for the dataset details, a total of 42,000 locations were randomly selected within the territories of Ethiopia, Malawi, and Nigeria. Corresponding daytime satellite images were collected from Google Earth as the training data, nighttime light intensity data is acquired from the National Oceanic and Atmospheric Administration (NOAA) as the labels, where the light intensity categorized into $8$ classes, as shown in Fig. \ref{fig:real}. A neural network model based on the ResNet-18 architecture is then trained on this data set. The local data in each satellite is assumed to be non-IID in each orbit, i.e., each satellite contains data with at most $2$ kinds of labels. Each satellite contains $\operatorname{Unif}(600, 750)$. Other configurations are the same as the second task.
\end{itemize}
We run the experiments on PyTorch version 2.0.0, CUDA version 11.7, and python version 3.9.

\subsubsection{Baselines}
We consider the following algorithms for comparison:
\begin{itemize}
	\item \fedmega: The proposed algorithm.
	\item \hlsgd: The distributed algorithm exploiting both GSLs and ISLs. For intra-orbit aggregations, only the adjacent satellites exchange their local models \cite{guo2022hybrid}.
	\item \fedisl: Vanilla \fedavg executed by LEO satellites \cite{razmi2022board} in one orbit.
\end{itemize}
\begin{figure}[!]
	\begin{subfigure}{0.24\textwidth}
		\centering
		\includegraphics[width=0.98\linewidth]{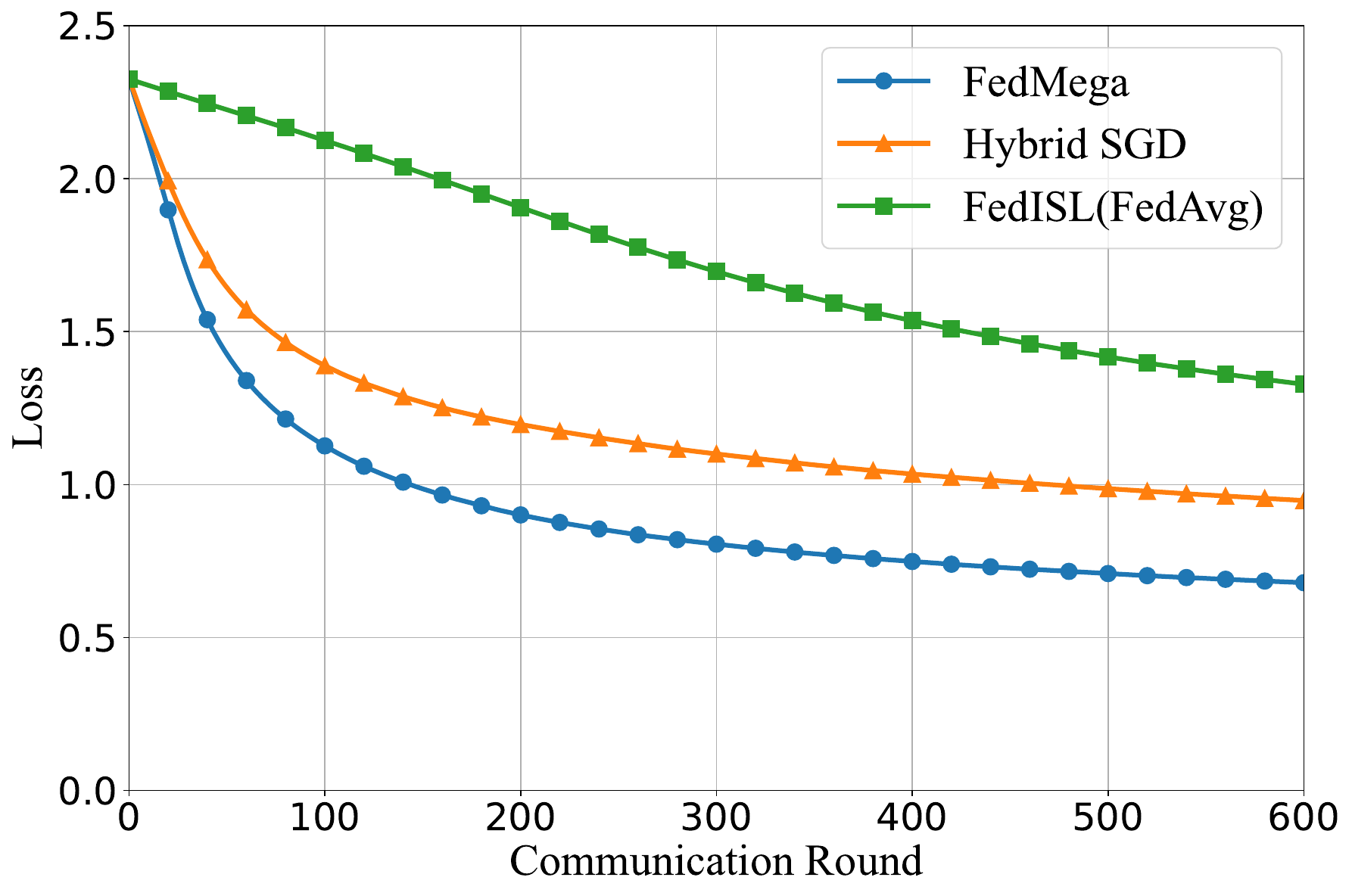}
		\caption{Training loss.}
		\label{fig:niidloss}
	\end{subfigure}
	\hfill
	\begin{subfigure}{0.24\textwidth}
		\centering
		\includegraphics[width=0.98\linewidth]{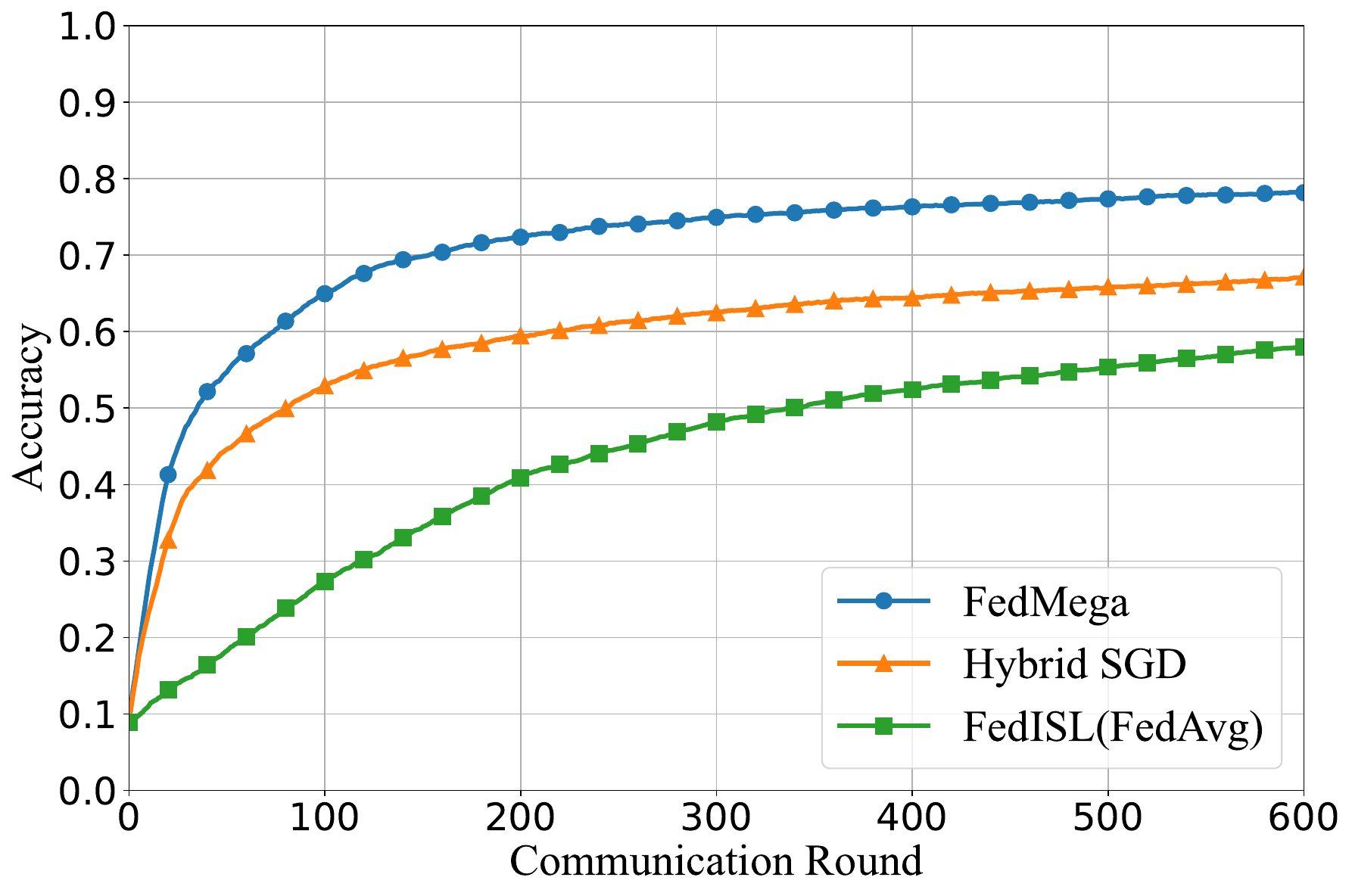}
		\caption{Test accuracy.}
		\label{fig:niidacc}
	\end{subfigure}
	\caption{Learning performance versus communication rounds of Synthetic dataset with non-IID data.}
	\label{fig:result1niid}
	\vspace{-0.5cm}
\end{figure}

\subsubsection{Delay Computation of Satellite FEEL}
We now present the delay computation for all baselines.
For \fedmega, in each communication round, we denote the delay of a single on-board local update, ring all-reduce based intra-orbit aggregation, and distributed downloading as $t_{\text{comp}}$, $t_{\text{rar}}$, and $t_{\text{down}}$, respectively. Besides, the time occupied for waiting a feasible GSL and global model broadcasting is denoted as $t_{\text{wait}}$ and $t_{\text{bc}}$. In this case, the total delay of communication round $r$ can be written as
\begin{align}\label{eq: tfedmega}
	T_{\fedmega}^r = 2 t^r_{\text{down}}  + t^r_{\text{bc}} + T (Et^r_{\text{comp}} + t^r_{\text{rar}}) + t^r_{\text{wait}},
\end{align}
where we assume that the downlink GSL and the uplink one are dual process.
To be specific, let $I$ be the model size of the model parameter (0.5GB for Synthetic data classification and 3GB for nighttime intensity prediction), $\gamma_{\text{gsl}}$ be the average data rate of the GSL, and $\gamma_{\text{isl}}$ be the average data rate of the laser ISL, we have
\begin{align}
	t_{\text{down}} = \frac{I}{\gamma_{\text{gsl}}},t_{\text{bc}} = \frac{I}{\gamma_{\text{isl}}},t_{\text{rar}} = \frac{2K-2}{2K}\frac{I}{\gamma_{\text{isl}}}+(2K-2)t_{\text{sum}},
\end{align}
where $t_{\text{sum}}$ is the time for performing a summation in Algorithm \ref{alg:ring-all-reduce}. 
Assume that, $\gamma_{\text{isl}} = 10 $ GB/s \cite{chaudhry2021laser}, $t^r_{\text{comp}}=2$ s for Synthetic data classification task and $5$ s for nighttime intensity prediction task, and $t_{\text{sum}} = 0.01$ s. The satellite-ground communication is modeled as a digital communication process over a complex Gaussian channel, where the link rate $ \gamma_{\text{gsl}}$ is computed based on the Shannon formula, i.e., 
$
    \gamma_{\text{gsl}}=B \log _2\left(1+\frac{HP_t G_t G_s}{B N_0}\right),
$
where $H = \left(\frac{c_0}{4\pi f_c d}\right)^2 $ is the free space loss, $d$ is the distance and $f_c = 32$ $\mathrm{GHz}$ is the carrier frequency. $P_t$ is the transmission power set to be $ 40$ $\mathrm{dBm}$, $G_t = 15$ $\mathrm{dBi}$ and $G_s = 30$ $\mathrm{dBi}$ are the transmitter and receiver antenna gains. $B$ is the bandwidth set to be $62.5$ $\mathrm{MHz}$. $N_0 = k_B T$ is the noise spectral density where $T = 354$ $\mathrm{K}$ is the temperature and $k_B$ is the Boltzman constant. The elevation angle is set to be $\pi/4$. The access time is the time required for establishing a connection between the satellite and the ground, set to be 10 seconds. 

\begin{figure}[!]
	\centering
	\begin{subfigure}{0.24\textwidth}
		\centering
		\includegraphics[width=1\linewidth]{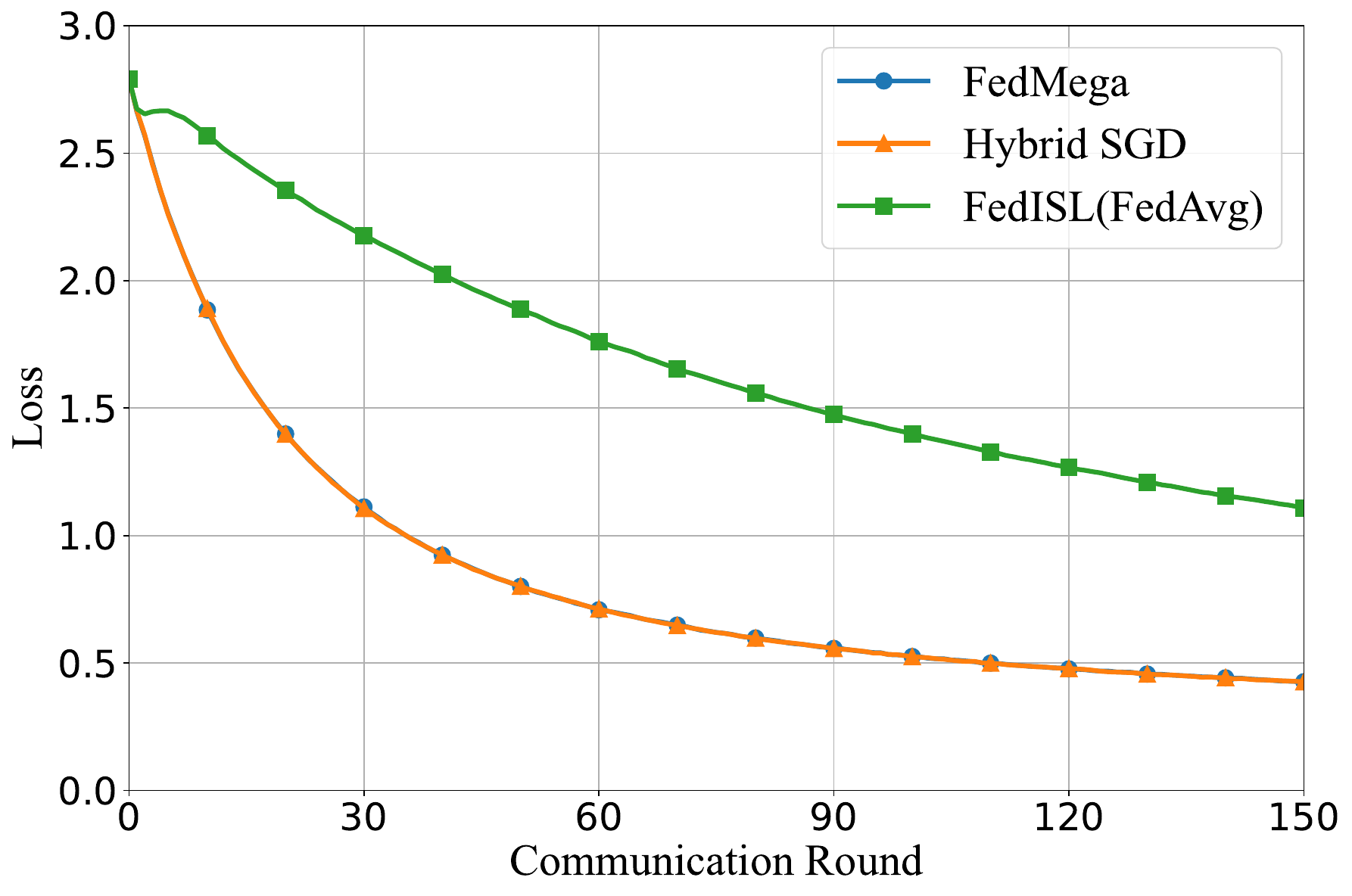}
		\caption{Training loss.}
		\label{fig:iidloss}
	\end{subfigure}
	\hfill
	\begin{subfigure}{0.24\textwidth}
		\centering
		\includegraphics[width=1\linewidth]{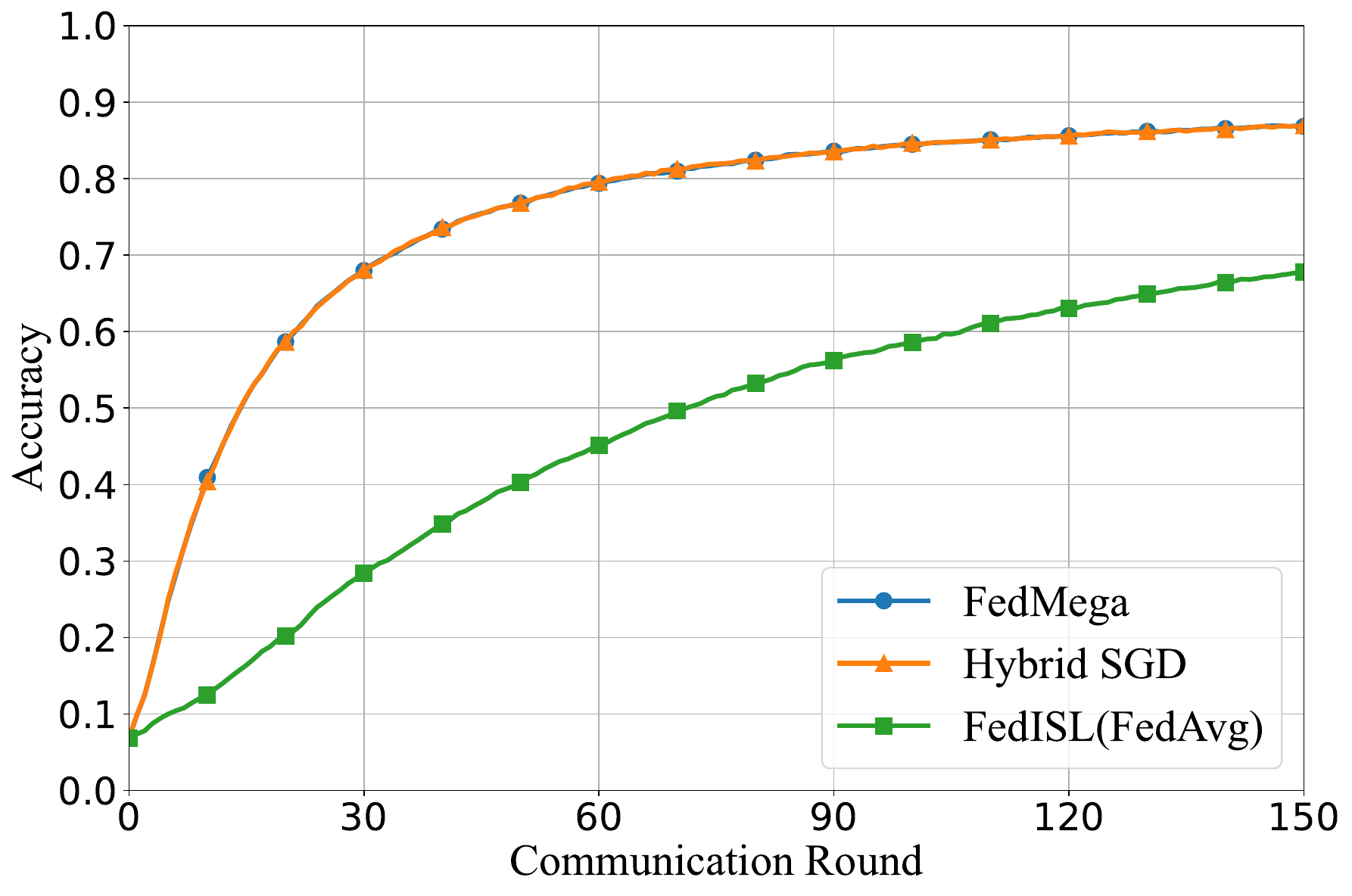}
		\caption{Test accuracy.}
		\label{fig:iidacc}
	\end{subfigure}
	\caption{Learning performance versus communication rounds of Land cover classification with IID data.}
	\label{fig:result1iid}
	\vspace{-0.5cm}
\end{figure}

For comparison, we also assume that the delay $t_{\text{down}}$, $t_{\text{bc}}$, $t_{\text{comp}}$, and $t_{\text{wait}}$ are the same for \hlsgd and \fedisl. 
The differences lie in the number of local updates and the  intra-orbit aggregations. The total delay of communication round $r$ for \hlsgd and \fedisl can be written as
\begin{align*}
	T_{\hlsgd}^r &= 2  t^r_{\text{down}}  + t^r_{\text{bc}} + T (Et^r_{\text{comp}} + t^r_{\text{hybrid}}) + t^r_{\text{wait}},\\
	T_{\fedisl}^r &= 2  t^r_{\text{down}}  + t^r_{\text{bc}} + Et^r_{\text{comp}} + t^r_{\text{isl}} + t^r_{\text{wait}},
\end{align*}
where $t^r_{\text{hybrid}}$ and $t^r_{\text{isl}}$ are the delay of intra-orbit model exchange proposed in \cite{guo2022hybrid} and model transmission in the orbit \cite{razmi2022board}, respectively.

\subsection{Performance Evaluation}
We perform the simulation of synthetic data classification task in the $300/6/1$ Walker Delta constellation and the nighttime light intensity prediction task in the $80/4/1$ Walker Delta constellation, respectively.

\subsubsection{Learning Performance versus Communication Rounds}
The training loss and the test accuracy versus communication rounds of Synthetic data classification are shown in Fig. \ref{fig:result1niid}.
It can be observed that with ISL equipped, both \fedmega and \hlsgd drastically outperform \fedavg without ISLs. 
As analyzed in Remark \ref{rem: algorithm comparison}, \fedmega converges faster and achieves higher accuracy than \hlsgd in the presence of intra-orbit heterogeneous data.
As illustrated in Fig. \ref{fig:result1iid}, in the land cover classification task with IID data splitting, \fedmega shares the same learning performance with \hlsgd. This means \hlsgd is not robust to non-IID data and \fedmega outperforms \hlsgd in the non-IID settings. 
Let's consider the land cover classification task for an example, the highest accuracy after 150 communication rounds for \fedmega, \hlsgd, and \fedisl is $86.91\%$, $86.85\%$, and $67.70\%$, respectively.
In all, the convergence speed-up of \fedmega is consistent with our analysis in Remark \ref{rem: algorithm comparison}.
\subsubsection{Learning Performance versus System Delay}
Based on the delay analysis in Section \ref{sec: setup}, the training loss and the test accuracy versus system delay of Synthetic data classification and nighttime intensity prediction are illustrated in Figs. \ref{fig:result2niidtime} and \ref{fig:realresult2niidtime}. From Figs. \ref{fig:result2niidtime} and \ref{fig:realresult2niidtime}, we can conclude that \fedmega has the smallest system delay to reach a targeted learning performance. For instance, in the task of Synthetic data classification, \fedmega demonstrates a significant delay reduction to reach test accuracy of $60\%$, up to $66.9\%$ and $85.1\%$ compared with \hlsgd and \fedavg, respectively. The delay reduction can be observed in the task of nighttime intensity prediction as well.

\begin{figure}[!]
	\begin{subfigure}{0.24\textwidth}
		\centering
		\includegraphics[width=1\linewidth]{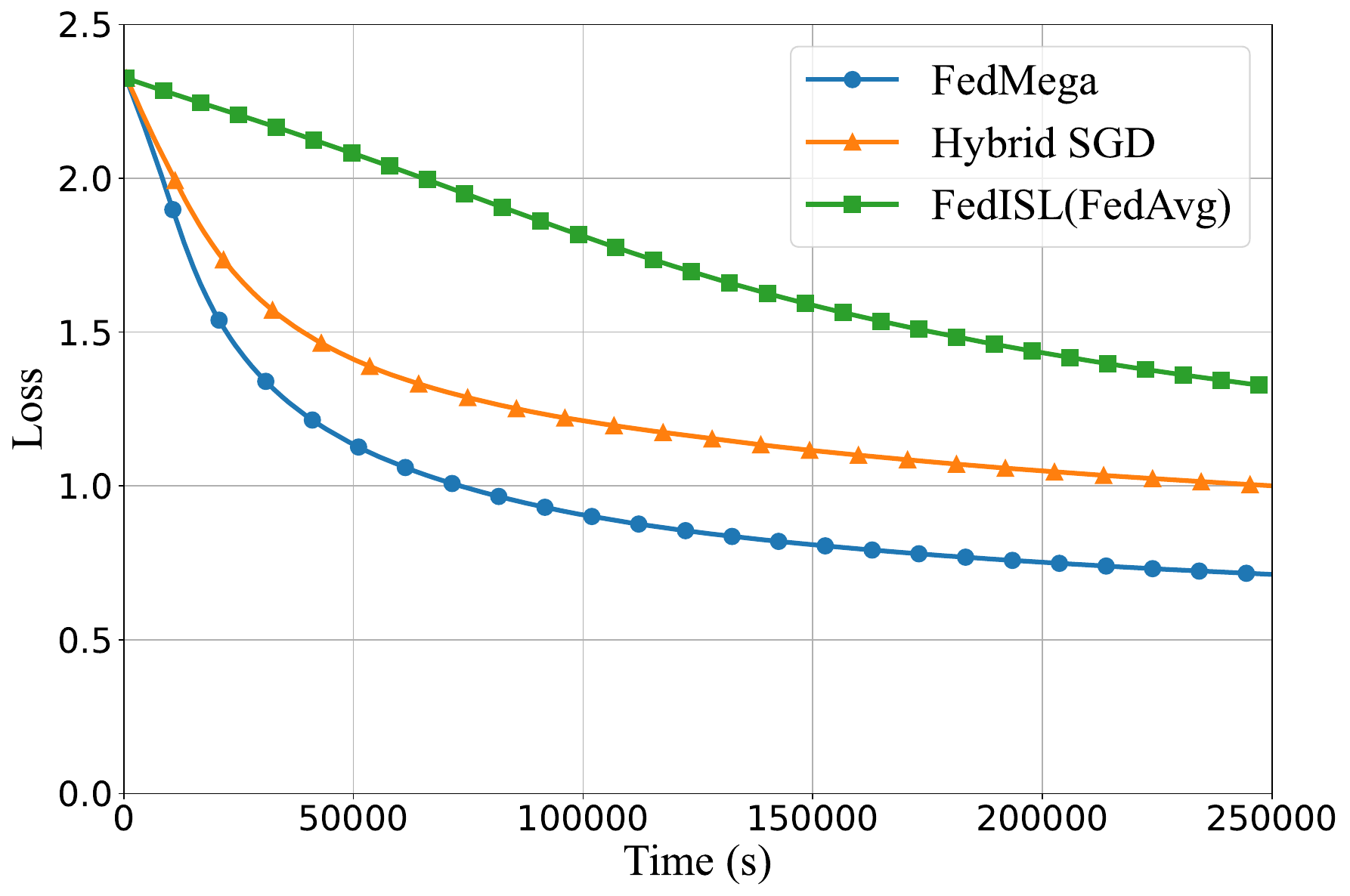}
		\caption{Training loss.}
		\label{fig:niidlosstime}
	\end{subfigure}
	\hfill
	\begin{subfigure}{0.24\textwidth}
		\centering
		\includegraphics[width=1\linewidth]{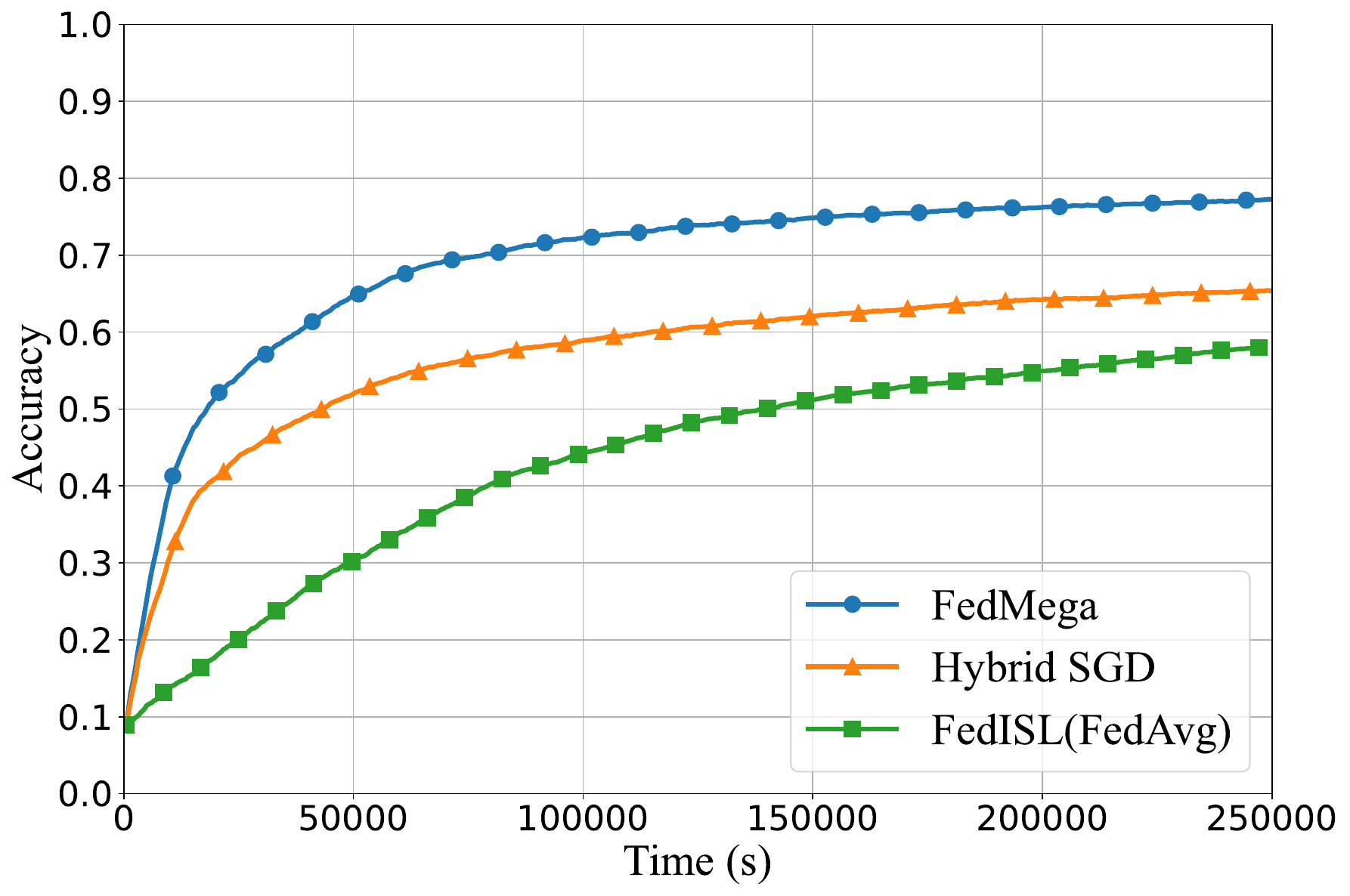}
		\caption{Test accuracy.}
		\label{fig:niidacctime}
	\end{subfigure}
	\caption{Learning performance versus system delay of Synthetic dataset with non-IID data.}
	\label{fig:result2niidtime}
	\vspace{-0.3cm}
\end{figure}



\subsubsection{Trade-off Analysis}
In Remark \ref{rem: convergence rate}, we state that the product $ET$ cannot be infinite large. Besides, we can also observe from \eqref{eq: tfedmega} that a large $ET$ will result in a large system delay of \fedmega. Therefore, there is a trade-off between frequent intra-orbit aggregations and the system latency for achieving a certain convergence accuracy. To verify this, we first fix the number of local updates, i.e., $E = 5$, and then derive the system delay to reach test accuracy of $65\%,70\%$, and $75\%$ under different choices of $T$ in the Synthetic data classification task, which is shown in Fig. \ref{fig:tradeoff}.

\begin{figure}[!]
	\begin{subfigure}{0.24\textwidth}
		\centering
		\includegraphics[width=1\linewidth]{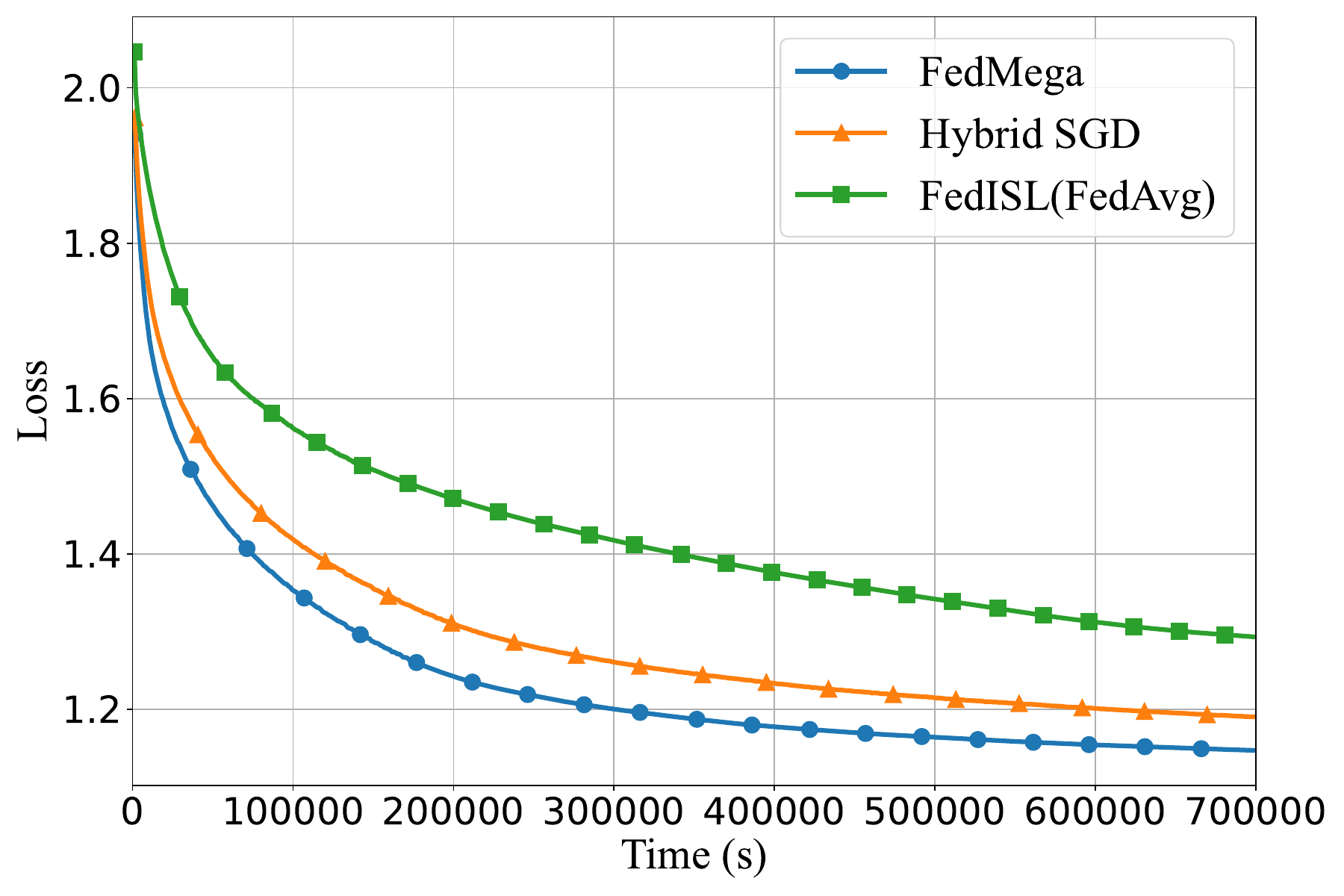}
		\caption{Training loss.}
		\label{fig:realniidlosstime}
	\end{subfigure}
	\hfill
	\begin{subfigure}{0.24\textwidth}
		\centering
		\includegraphics[width=1\linewidth]{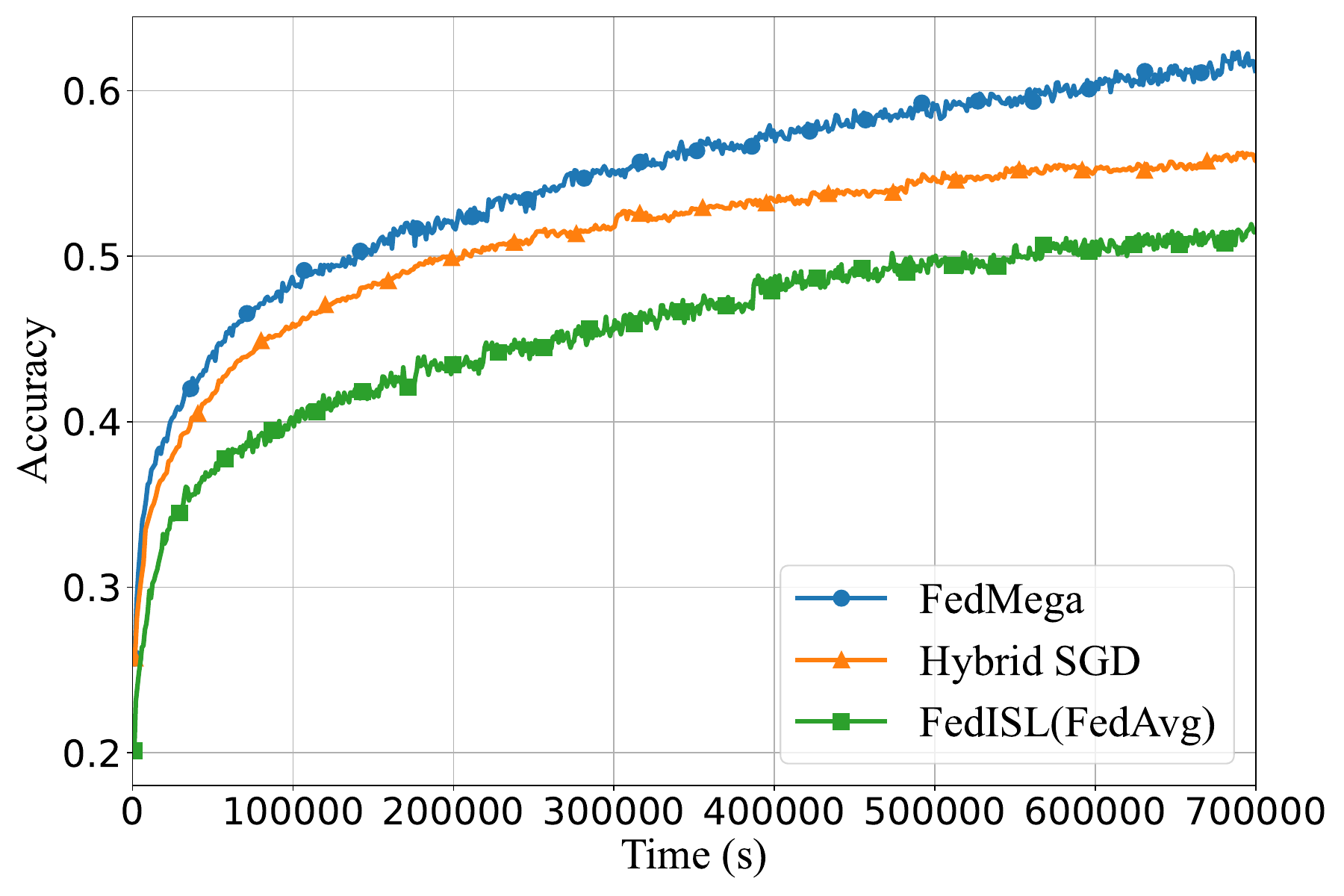}
		\caption{Test accuracy.}
		\label{fig:realniidacctime}
	\end{subfigure}
	\caption{Learning performance versus system delay of nighttime intensity prediction with non-IID data.}
	\label{fig:realresult2niidtime}
	\vspace{-0.5cm}
\end{figure}

From Fig. \ref{fig:tradeoff}, we can summarize that the number of intra-orbit aggregations $T$ can be neither too smaller, e.g., $T=1$, nor too large, e.g., $T=100$. When $T = 30$, \fedmega has the smallest system delay to arrive at different levels of learning performance, which is definitely a proper design of $T$. 

\begin{figure}[!]
	\centering
	\includegraphics[width=0.9\linewidth]{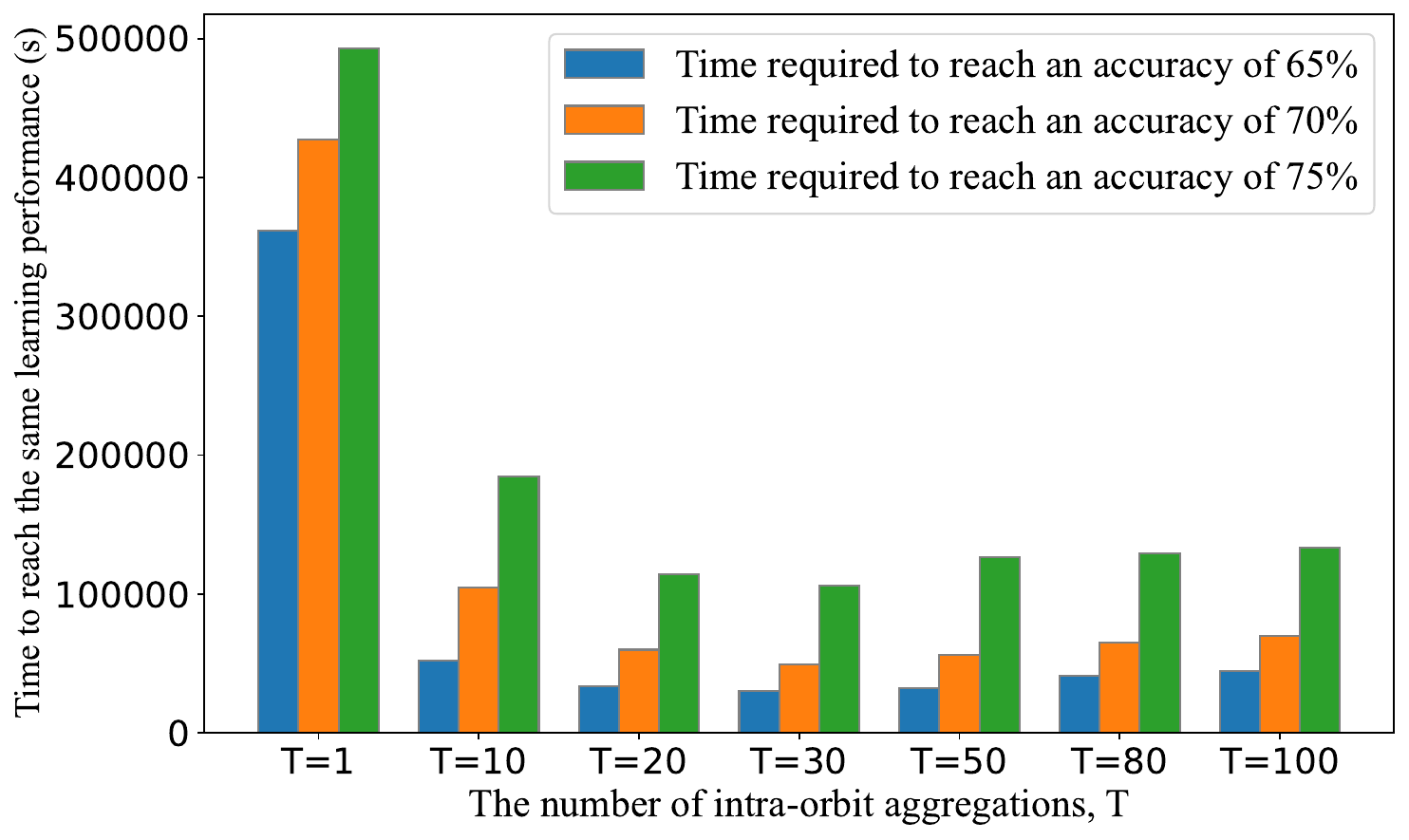}
	\caption{Trade-off between the number of intra-orbit aggregations and the total delay to reach a target accuracy.}
	\label{fig:tradeoff}
	\vspace{-0.5cm}
\end{figure}

	\section{Conclusion}

In this paper, we developed a fast-converging and communication-efficient satellite FEEL framework for LEO networks. Recognizing that the extensive use of GSL significantly hinders the convergence speed of satellite FEEL, we leveraged the high data rate of laser ISL along with the stable ring topology within each orbit. This approach introduces an intra-orbit model aggregation step, which substantially reduces GSL dependency and, consequently, reduces latency. To further accelerate the intra-orbit model aggregation, we employed a ring all-reduce-based transmission scheme, ensuring execution within a constant and short time irrelevant to the number of satellites. Additionally, we devised a network flow-based transmission scheme to further minimize the latency of the global model aggregation step. Comprehensive convergence analysis under non-convex settings and non-IID data distribution characterized the superior convergence performance of our proposed algorithm. Simulation results confirmed its superior performance compared to existing benchmarks.
\begin{figure*}[ht] 
	\centering 
       
	\begin{align}\nonumber\label{eq: T3}
		&\frac{1}{M}\sum_{\mathcal{M}}\mathbb{E}\left[\left \|\bar{\bm{z}}^{r,t}- \bar{\bm{z}}_m^{r,t} \right \|_2^2\right]
		=\frac{\eta^2E^2}{M}\sum_{\mathcal{M}}\mathbb{E}\Big[\Big \| \sum_{\tau=0}^{t-1}\left[\bar{\bm{d}}_m^{r,\tau} \pm \nabla F_m(\bar{\bm{z}}_m^{r,\tau})\pm \nabla F(\bar{\bm{z}}_j^{r,\tau}) -\bar{\bm{d}}^{r,\tau}\right] \Big \|_2^2\Big]\\\nonumber
		\overset{(a)}{\leq}&\frac{2\eta^2E^2}{M}\sum_{\mathcal{M}}\mathbb{E}\Big[\Big \| \sum_{\tau=0}^{t-1}\left[\bar{\bm{d}}_m^{r,\tau} - \nabla F_m(\bar{\bm{z}}_m^{r,\tau})\right]- \left[\bar{\bm{d}}^{r,\tau} -\nabla F(\bar{\bm{z}}_j^{r,\tau})\right]  \Big \|_2^2\Big]
		+\frac{2\eta^2E^2}{M}\sum_{\mathcal{M}}\mathbb{E}\Big[\Big \| \sum_{\tau=0}^{t-1}\left[\nabla F_m(\bar{\bm{z}}_m^{r,\tau})- \nabla F(\bar{\bm{z}}_j^{r,\tau})\right]  \Big \|_2^2\Big]\\\nonumber
		\overset{(b)}{\leq}&\underbrace{\frac{4\eta^2E^2}{M}\sum_{\mathcal{M}}\mathbb{E}\Big[\Big \| \sum_{\tau=0}^{t-1}\left[\bar{\bm{d}}_m^{r,\tau}-\bar{\bm{h}}_m^{r,\tau}\right] - \left[\bar{\bm{d}}^{r,\tau}-\bar{\bm{h}}^{r,\tau}\right] \Big \|_2^2\Big]}_{T_{31}}
		+\underbrace{\frac{4\eta^2E^2}{M}\sum_{\mathcal{M}}\mathbb{E}\Big[\Big \| \sum_{\tau=0}^{t-1}\left[\bar{\bm{h}}_m^{r,\tau} -\nabla F_m(\bar{\bm{z}}_m^{r,\tau})\right] -\left[\bar{\bm{h}}^{r,\tau}-\nabla F(\bar{\bm{z}}_j^{r,\tau})\right]  \Big \|_2^2\Big]}_{T_{32}}\\
		&+\underbrace{\frac{2\eta^2E^2}{M}\sum_{\mathcal{M}}\mathbb{E}\Big[\Big \| \sum_{\tau=0}^{t-1}\left[\nabla F_m(\bar{\bm{z}}_m^{r,\tau})- \nabla F(\bar{\bm{z}}_j^{r,\tau})\right]  \Big \|_2^2\Big]}_{T_{33}},
	\end{align}
	\vspace*{8pt} 
	\hrulefill 
	\vspace{-0.5cm}
\end{figure*}
\begin{figure*}
	\begin{flalign}\nonumber\label{eq: T31}
		T_{31} \overset{(a)}{=} &\frac{4\eta^2E^2}{M}\sum_{\mathcal{M}}\mathbb{E}\Big[\Big \| \sum_{\tau=0}^{t-1}\left[\bar{\bm{d}}_m^{r,\tau}-\bar{\bm{h}}_m^{r,\tau}\right] \Big \|_2^2\Big]
		-\frac{4\eta^2E^2}{M}\sum_{\mathcal{M}}\mathbb{E}\Big[\Big \| \sum_{\tau=0}^{t-1} \left[\bar{\bm{d}}^{r,\tau}-\bar{\bm{h}}^{r,\tau}\right] \Big \|_2^2\Big]&\\\nonumber
		\overset{(b)}{=}& \frac{4\eta^2E^2}{M}\sum_{\tau=0}^{t-1}\sum_{\mathcal{M}}\mathbb{E}\left[\left \| \bar{\bm{d}}_m^{r,\tau}-\bar{\bm{h}}_m^{r,\tau} \right \|_2^2\right]
		-4\eta^2E^2\sum_{\tau=0}^{t-1}\mathbb{E}\Big[\Big \| \bar{\bm{d}}^{r,\tau}-\bar{\bm{h}}^{r,\tau} \Big \|_2^2\Big]\\
		\overset{(c)}{\leq}& \frac{4\eta^2E(M-1)\sigma^2t}{K},\hfill
	\end{flalign}
	\vspace*{8pt} 
	\hrulefill
	\begin{flalign}\nonumber\label{eq: T32}
		T_{32}
		\overset{(a)}{\leq} & \frac{4\eta^2E^2}{M}\sum_{\mathcal{M}}\mathbb{E}\Big[\Big \| \sum_{\tau=0}^{t-1}\left[\bar{\bm{h}}_m^{r,\tau} -\nabla F_m(\bar{\bm{z}}_m^{r,\tau})\right]   \Big \|_2^2\Big]
		\overset{(b)}{\leq} \frac{4\eta^2E^2t}{M}\sum_{\tau=0}^{t-1}\sum_{\mathcal{M}}\mathbb{E}\Big[\Big \| \frac{1}{K_0E}\sum_{\mathcal{K}_m}\sum_{e=0}^{E-1}\left[\nabla F_{m,k}(\bm{z}_{m,k}^{r,\tau,e})-\nabla F_{m,k}(\bar{\bm{z}}_{m}^{r,\tau})\right]  \Big \|_2^2\Big]\\
		\overset{(c)}{\leq}& \frac{4\eta^2EL^2t}{K}\sum_{\tau=0}^{t-1}\sum_{\mathcal{K}}\sum_{e=0}^{E-1}\mathbb{E}\left[\left \| \bm{z}_{m,k}^{r,\tau,e}-\bar{\bm{z}}_{m}^{r,\tau}\right \|_2^2\right]
	\end{flalign}
	\vspace*{8pt} 
	\hrulefill
	\begin{flalign}\nonumber\label{eq: T3'}
		T_3
		\overset{(a)}{\leq}& \eta E L^2 \left\{\frac{4\eta^2E(M-1)\sigma^2}{K}\sum_{t=0}^{T-1}t+\frac{4\eta^2EL^2}{K}\sum_{t=0}^{T-1}t\sum_{\tau=0}^{T-1}\sum_{\mathcal{K}}\sum_{e=0}^{E-1}\mathbb{E}\left[\left \| \bm{z}_{m,k}^{r,\tau,e}-\bar{\bm{z}}_{m}^{r,\tau}\right \|_2^2\right]\right.\\\nonumber
		\phantom{=\;\;}& + \left. \frac{4\eta^2E^2L^2}{K}\sum_{t=0}^{T-1}t\sum_{\tau=0}^{T-1}\sum_{\mathcal{M}}\mathbb{E}\left[\left \| \bar{\bm{z}}^{r,\tau}-\bar{\bm{z}}_{m}^{r,\tau}\right \|_2^2\right]+ 4\eta^2E^2\sum_{t=0}^{T-1}t\sum_{\tau=0}^{T-1}\left[\alpha \left  \| \nabla F(\bar{\bm{z}}^{r,\tau}) \right \|_2^2 + \beta \right]\right\}&\\\nonumber
		\overset{(b)}{\leq}& \frac{\eta L^2(M-1)\sigma^2D_T}{K(1-D_T)}+\frac{\eta L^2D_T}{K(1-D_T)}\sum_{\tau=0}^{T-1}\sum_{\mathcal{K}}\sum_{e=0}^{E-1}\mathbb{E}\left[\left \| \bm{z}_{m,k}^{r,\tau,e}-\bar{\bm{z}}_{m}^{r,\tau}\right \|_2^2\right] +\frac{\eta E D_T}{1-D_T}\sum_{\tau=0}^{T-1}\left[\alpha \left  \| \nabla F(\bar{\bm{z}}^{r,\tau}) \right \|_2^2 + \beta \right]\\
		= & \frac{\eta L^2(M-1)\sigma^2D_T}{K(1-D_T)}+\frac{D_T}{(1-D_T)}T_4
		+ \frac{\eta E D_T}{1-D_T}\sum_{\tau=0}^{T-1}\left[\alpha \left  \| \nabla F(\bar{\bm{z}}^{r,\tau}) \right \|_2^2 + \beta \right],
	\end{flalign}
	\vspace*{8pt} 
	\hrulefill 
	\vspace{-0.5cm}
\end{figure*}

\appendix
\setlength{\jot}{-0.1em}
\allowdisplaybreaks
\section{Proof of Theorem \ref{thm: theorem 1}}\label{sec: proof of theorem 1} 
\subsection{Preliminaries}
Firstly, we define the following notations for abbreviation.
\begin{itemize}
	\item We use $\sum_{\mathcal{K}}$ to denote the abbreviation for $\sum_{m=1}^M \sum_{k=1}^{K_0}$.
	\item We use $\sum_{\mathcal{K}_m}$ to denote the abbreviation for $\sum_{k=1}^{K_0}$.
	\item We use $\sum_{\mathcal{M}}$ to denote the abbreviation for $\sum_{m=1}^{M}$.
	\item We use $a\pm b$ to denote the abbreviation for $a+b-b$.
	\item We use $\bm{d}_{m,k}^{r,t}=\frac{1}{E}\sum_{e=0}^{E-1}\nabla F_{m,k}^{r,t,e}(\bm{z}_{m,k}^{r,t,e})$ to represent the average stochastic gradient of each satellite $\mathsf{S}_{m, k}$ between two intra-orbit aggregations, and $\bm{h}_{m,k}^{r,t}=\frac{1}{E}\sum_{e=0}^{E-1}\nabla F_{m,k}(\bm{z}_{m,k}^{r,t,e})$ to represent the average true gradient likewise.
	\item We define $\bar{\bm{d}}^{r,t} = \frac{1}{K}\sum_{\mathcal{K}} \bm{d}_{m,k}^{r,t}$ and $\bar{\bm{h}}^{r,t} = \frac{1}{K}\sum_{\mathcal{K}} \bm{h}_{m,k}^{r,t}$ as the average stochastic and true gradients of all satellites.
	\item We define $\bar{\bm{d}}_m^{r,t} = \frac{1}{K_0}\sum_{\mathcal{K}_m} \bm{d}_{m,k}^{r,t}$ and $\bar{\bm{h}}_m^{r,t} = \frac{1}{K_0}\sum_{\mathcal{K}_m} \bm{h}_{m,k}^{r,t}$ as the average stochastic and true gradients of all satellites in one orbit.
\end{itemize}

Secondly, we have the following Lemmas to assist the proof.
\begin{lem}\label{lem: Jensen}
	Jensen's Inequality. With $w_n\in [0,1], \sum_{n=1}^{N}w_n =1$, it holds that
	\begin{equation*}
		\left\|\sum_{n=1}^{N}w_n\bm x_n\right\|_2^2\leq\sum_{n=1}^{N}w_n\left\|\bm x_n\right\|_2^2.
	\end{equation*}
\end{lem}

\begin{lem}\label{lem: lemma1}
	For $w_n\in [0,1], \sum_{n=1}^{N}w_n =1$, we have

	\begin{equation*}
		\sum_{n=1}^{N}w_n\left\|\bm x_n-\bar{\bm{x}}\right\|_2^2=\sum_{n=1}^{N}w_n\left\|\bm x_n\right\|_2^2 - \|\bar{\bm{x}}\|_2^2\leq \sum_{n=1}^{N}w_n\left\|\bm x_n\right\|_2^2.
	\end{equation*}
\end{lem}

\begin{lem}\label{lem: lemma2}
	Suppose  $ \left\{A_{k}\right\}_{k = 1}^{T} $  is a sequence of random matrices and  $ \mathbb{E}\left[A_{k} \mid A_{k-1}, A_{k-2}, \ldots, A_{1}\right] = \mathbf{0}, \forall k $. Then,
	\begin{align*}
		\mathbb{E}\left[\left\|\sum_{k  = 1}^{T} A_{k}\right\|_{F}^{2}\right]  = \sum_{k = 1}^{T} \mathbb{E}\left[\left\|A_{k}\right\|_{F}^{2}\right].
	\end{align*}
\end{lem}

\subsection{Proof of Theorem}
Based on the intra-orbit updating rules, obviously, we have
\vspace{-0.3cm}
\begin{align}\nonumber\label{eq: expectation1}
	&\mathbb{E}\left[F(\bar{\bm{z}}^{r,t+1}) - F(\bar{\bm{z}}^{r,t})\right]\\\nonumber
	\leq& -\eta E\mathbb{E} \left[\left <\nabla  F(\bar{\bm{z}}^{r,t}),\bar{\bm{d}}^{r,t} \right>\right]
	+\frac{\eta^2E^2L}{2}\mathbb{E}\left[\left \|\bar{\bm{d}}^{r,t} \right \|_2^2\right]\\\nonumber
	=&-\eta E\mathbb{E} \left[\left <\nabla  F(\bar{\bm{z}}^{r,t}),\bar{\bm{h}}^{r,t} \right>\right]
	+\frac{\eta^2E^2L}{2}\mathbb{E}\left[\left \|\bar{\bm{d}}^{r,t} \right \|_2^2\right]\\\nonumber
	\overset{(a)}{=}&-\frac{\eta E }{2}\left \|\nabla  F(\bar{\bm{z}}^{r,t})\right \|_2^2 -\frac{\eta E }{2}\mathbb{E}\left[\left \| \bar{\bm{h}}^{r,t} \right \|_2^2\right]\\
	+& \underbrace{\frac{\eta E }{2}\mathbb{E}\left[\left \|\nabla  F(\bar{\bm{z}}^{r,t})-\bar{\bm{h}}^{r,t}  \right \|_2^2\right] }_{T_1}
	+\underbrace{\frac{\eta^2E^2L}{2}\mathbb{E}\left[\left \|\bar{\bm{d}}^{r,t} \right \|_2^2\right]}_{T_2},
\end{align}
where the expectation is taken over local mini-batch sampling, $ (a) $ holds by the fact that $2\left<\bm a,\bm b\right> = \|\bm a\|_2^2+\|\bm b\|_2^2-\|\bm a-\bm b\|_2^2$.
To give an upper bound of $T_1$, we have
\begin{align}\nonumber\label{eq: T1}
	T_1
	= &\frac{\eta E }{2}\mathbb{E}\Big[\Big \|\nabla  F(\bar{\bm{z}}^{r,t})-\bar{\bm{h}}^{r,t}\pm \frac{1}{M} \sum_{\mathcal{M}}\nabla  F_m(\bar{\bm{z}}_m^{r,t}) \Big \|_2^2\Big] \\\nonumber
	\overset{(a)}{\leq }&\eta E\mathbb{E}\Big[\Big \|\nabla  F(\bar{\bm{z}}^{r,t})- \frac{1}{M} \sum_{\mathcal{M}}\nabla  F_m(\bar{\bm{z}}_m^{r,t}) \Big \|_2^2\Big] \\\nonumber
	+ &\eta E \mathbb{E}\Big[\Big \|\bar{\bm{h}}^{r,t}- \frac{1}{M} \sum_{\mathcal{M}}\nabla  F_m(\bar{\bm{z}}_m^{r,t}) \Big \|_2^2\Big]\\\nonumber
	\overset{(b)}{\leq }& \frac{\eta E}{M} \sum_{\mathcal{M}}\mathbb{E}\Big[\Big\|\nabla  F(\bar{\bm{z}}^{r,t})- \nabla  F_m(\bar{\bm{z}}_m^{r,t}) \Big \|_2^2\Big] \\\nonumber
	+ &\frac{\eta E}{K} \sum_{\mathcal{K}} \mathbb{E}\Big[\Big \|\frac{1}{E} \sum_{e=0}^{E-1}\left[\nabla F_{m,k}^{r,t,e}(\bm{z}_{m,k}^{r,t,e})- \nabla  F_{m,k}(\bar{\bm{z}}_m^{r,t})\right] \Big \|_2^2\Big]\\\nonumber
	\overset{(c)}{\leq }& \frac{\eta EL^2}{M} \sum_{\mathcal{M}}\mathbb{E}\left[\left \|\bar{\bm{z}}^{r,t}- \bar{\bm{z}}_m^{r,t} \right \|_2^2\right] \\
	+ &\frac{\eta L^2}{K} \sum_{\mathcal{K}}\sum_{e=0}^{E-1} \mathbb{E}\left[\left \|\bm{z}_{m,k}^{r,t,e}- \bar{\bm{z}}_m^{r,t}\right \|_2^2\right],
\end{align}
where $ (a) $ holds due to the fact that $\|\bm a+ \bm b\|_2^2\leq2\|\bm a\|_2^2+2\|\bm b\|_2^2$, $ (b) $ and $ (c) $ hold by Lemma \ref{lem: Jensen} and Assumption \ref{ass: smooth}.
To give an upper bound of $T_2$, according to \cite[Lemma 2 of Appendix]{zhu2024over}, we have
\begin{align}\label{eq: T2}
	T_{2}\leq \frac{\eta^2EL\sigma^2}{K}+\eta^2E^2L \mathbb{E}\left[\left \|\bar{\bm{h}}^{r,t} \right \|_2^2\right].
\end{align}
By combining the results obtained in \eqref{eq: expectation1}, \eqref{eq: T1}, and \eqref{eq: T2}, we have
\begin{align}\nonumber\label{eq: expectation2}
	&\mathbb{E}\left[F(\bar{\bm{z}}^{r,t+1}) - F(\bar{\bm{z}}^{r,t})\right]\\\nonumber
	\overset{(a)}{\leq}&-\frac{\eta E }{2}\left \|\nabla  F(\bar{\bm{z}}^{r,t})\right \|_2^2 +\frac{\eta EL^2}{M} \sum_{\mathcal{M}}\mathbb{E}\left[\left \|\bar{\bm{z}}^{r,t}- \bar{\bm{z}}_m^{r,t} \right \|_2^2\right] \\
	& +\frac{\eta L^2}{K} \sum_{\mathcal{K}}\sum_{e=0}^{E-1} \mathbb{E}\left[\left \|\bm{z}_{m,k}^{r,t,e}- \bar{\bm{z}}_m^{r,t}\right \|_2^2\right]+\frac{\eta^2EL\sigma^2}{K},
\end{align}
where $ (a) $ holds by $\eta\leq \frac{1}{2EL}$.
By summing \eqref{eq: expectation2} from $t=0$ to $T-1$, we have
\begin{align}\nonumber\label{eq: expectation3}
	&\mathbb{E}\left[F(\bar{\bm{z}}^{r+1}) - F(\bar{\bm{z}}^{r})\right]&\\\nonumber
	\leq&-\frac{\eta E }{2}\sum_{t=0}^{T-1}\left \|\nabla  F(\bar{\bm{z}}^{r,t})\right \|_2^2 
	+\underbrace{\frac{\eta EL^2}{M} \sum_{t=0}^{T-1}\sum_{\mathcal{M}}\mathbb{E}\left[\left \|\bar{\bm{z}}^{r,t}- \bar{\bm{z}}_m^{r,t} \right \|_2^2\right]}_{T_3}& \\
	+& \underbrace{\frac{\eta L^2}{K}\sum_{t=0}^{T-1} \sum_{\mathcal{K}}\sum_{e=0}^{E-1} \mathbb{E}\left[\left \|\bm{z}_{m,k}^{r,t,e}- \bar{\bm{z}}_m^{r,t}\right \|_2^2\right]}_{T_4}
	+\frac{\eta^2EL\sigma^2T}{K}.&
\end{align}

To further give an upper bound of $T_3$, we first have Eq. \eqref{eq: T3}, 	where $ (a) $, $ (b) $ hold due to $\|\bm a+ \bm b\|_2^2\leq2\|\bm a\|_2^2+2\|\bm b\|_2^2$.

In addition, for $T_{31}$ in \eqref{eq: T3}, we have Eq. \eqref{eq: T31},
where $ (a), (b), (c) $ hold by Lemma \ref{lem: lemma1}, \ref{lem: lemma2} and Assumption \ref{ass: gradient variance} respectively.

For $T_{32}$ in \eqref{eq: T3}, we have Eq. \eqref{eq: T32},
where $ (a) $ is due to Lemmas \ref{lem: Jensen} and \ref{lem: lemma1}, $ (b)  $ and $ (c) $ hold by Lemma \ref{lem: Jensen} and Assumption \ref{ass: smooth}.

For $T_{33}$ in \eqref{eq: T3}, we have Eq. \eqref{eq: T33},
\begin{flalign}\nonumber\label{eq: T33}
	T_{33}
	\overset{(a)}{\leq} & \frac{2\eta^2E^2}{M}\sum_{\mathcal{M}}\mathbb{E}\Big[\Big \| \sum_{\tau=0}^{t-1}\left[\nabla F_m(\bar{\bm{z}}_m^{r,\tau})\pm \nabla F_m(\bar{\bm{z}}^{r,\tau})\right]  \Big \|_2^2\Big]&\\\nonumber
	\overset{(b)}{\leq}  & \frac{4\eta^2E^2t}{M}\sum_{\tau=0}^{t-1}\sum_{\mathcal{M}}\mathbb{E}\left[\left \| \nabla F_m(\bar{\bm{z}}_m^{r,\tau})- \nabla F_m(\bar{\bm{z}}^{r,\tau}) \right \|_2^2\right] \\\nonumber
	&+ \frac{4\eta^2E^2t}{M}\sum_{\tau=0}^{t-1}\sum_{\mathcal{M}}\mathbb{E}\left[\left \| \nabla F_m(\bar{\bm{z}}^{r,\tau}) \right \|_2^2\right]\\\nonumber
	\overset{(c)}{\leq}& \frac{4\eta^2E^2L^2t}{K}\sum_{\tau=0}^{t-1}\sum_{\mathcal{M}}\mathbb{E}\left[\left \| \bar{\bm{z}}^{r,\tau}-\bar{\bm{z}}_{m}^{r,\tau}\right \|_2^2\right] \\
	&+ 4\eta^2E^2t\sum_{\tau=0}^{t-1}\left[\alpha \left  \| \nabla F(\bar{\bm{z}}^{r,\tau}) \right \|_2^2 + \beta \right],
\end{flalign}
where $ (a) $ holds due to Lemma \ref{lem: lemma1}, $ (b) $ holds by Lemma \ref{lem: Jensen}, and $ (c) $ is due to Assumptions \ref{ass: smooth} and \ref{ass: inter-orbit dissimilarity}.
Based on the results in \eqref{eq: T31}, \eqref{eq: T32}, and \eqref{eq: T33}, we further give an upper bound of $T_3$ in Eq. \eqref{eq: T3'},
where $(a)$ holds due to the fact that $t\leq T-1$, $ (b) $ holds by setting $D_T = 2\eta^2E^2L^2T(T-1)$, and $\eta<\frac{1}{LE\sqrt{2T(T-1)}}\leq \frac{1}{2EL},T\geq2$.
Next, to give an upper bound of $T_4$, we first have
\begin{align}\nonumber\label{eq: T4}
	&\mathbb{E}\left[\left \|\bm{z}_{m,k}^{r,t,e}- \bar{\bm{z}}_m^{r,t} \right \|_2^2\right]\\\nonumber
	= & \eta^2 \mathbb{E}\Big[\Big \| \sum_{i = 0}^{e-1} \left[\nabla F_{m,k}^{r,t,i}(\bm{z}_{m,k}^{r,t,i})\pm  \nabla F_{m,k}(\bm{z}_{m,k}^{r,t,i})\right] \Big \|_2^2\Big]\\\nonumber
	\overset{(a)}{\leq} & 2\eta^2 \mathbb{E}\Big[\Big \| \sum_{i = 0}^{e-1} \left[\nabla F_{m,k}^{r,t,i}(\bm{z}_{m,k}^{r,t,i})-  \nabla F_{m,k}(\bm{z}_{m,k}^{r,t,i})\right] \Big \|_2^2\Big]\\\nonumber
	+ &2\eta^2 \mathbb{E}\left[\left \| \sum_{i = 0}^{e-1} \left[ \nabla F_{m,k}(\bm{z}_{m,k}^{r,t,i})\pm\nabla F_{m,k}(\bar{\bm{z}}_{m}^{r,t}) \right] \right \|_2^2\right]\\\nonumber
	\overset{(b)}{\leq} & 2\eta^2\sigma^2e + 4\eta^2L^2e \sum_{i = 0}^{e-1}  \mathbb{E}\left[\left \| \bar{\bm{z}}_{m}^{r,t}- \bm{z}_{m,k}^{r,t,i}\right \|_2^2\right] \\\nonumber
	+&4\eta^2e \sum_{i = 0}^{e-1} \mathbb{E}\left[\left \| \nabla F_{m,k}(\bar{\bm{z}}_{m}^{r,t}) \pm \nabla F_{m}(\bar{\bm{z}}_{m}^{r,t}) \pm \nabla F_{m}(\bar{\bm{z}}^{r,t})\right \|_2^2\right]\\\nonumber
	\overset{(c)}{\leq} & 2\eta^2\sigma^2e + 4\eta^2L^2e \sum_{i = 0}^{e-1}  \mathbb{E}\left[\left \| \bar{\bm{z}}_{m}^{r,t}- \bm{z}_{m,k}^{r,t,i}\right \|_2^2\right] \\\nonumber
	+& 12\eta^2e \sum_{i = 0}^{e-1} \left\{\mathbb{E}\left[\left \| \nabla F_{m,k}(\bar{\bm{z}}_{m}^{r,t}) - \nabla F_{m}(\bar{\bm{z}}_{m}^{r,t}) \right \|_2^2\right]\right.\\
	 + &\left.
	\mathbb{E}\left[\left \|\nabla F_{m}(\bar{\bm{z}}_{m}^{r,t}) - \nabla F_{m}(\bar{\bm{z}}^{r,t})\right \|_2^2\right] +\mathbb{E}\left[\left \|\nabla F_{m}(\bar{\bm{z}}^{r,t})\right \|_2^2\right]  \right\},
\end{align}
where $ (a) $, $ (b) $, and $(c)$ hold due to Lemma \ref{lem: Jensen} and Assumption \ref{ass: smooth}.
Based on \eqref{eq: T4}, it yields that
\begin{align}\nonumber\label{eq: T41}
	&\sum_{\mathcal{K}}\sum_{e=0}^{E-1}\mathbb{E}\left[\left \|\bm{z}_{m,k}^{r,t,e}- \bar{\bm{z}}_m^{r,t} \right \|_2^2\right]\\\nonumber
	\overset{(a)}{\leq} & 2\eta^2\sigma^2K\sum_{e=0}^{E-1}e + 4\eta^2L^2\sum_{\mathcal{K}}\sum_{e=0}^{E-1}e \sum_{i = 0}^{E-1}  \mathbb{E}\left[\left \| \bar{\bm{z}}_{m}^{r,t}- \bm{z}_{m,k}^{r,t,i}\right \|_2^2\right]\\\nonumber
	 +& 12\eta^2E\sum_{e=0}^{E-1}e  \left\{\sum_{\mathcal{M}}K_0\delta_m^2\right.\\\nonumber
	+& \left.
	K_0L^2\sum_{\mathcal{M}}\mathbb{E}\left[\left \|\bar{\bm{z}}_{m}^{r,t}- \bar{\bm{z}}^{r,t}\right \|_2^2\right] +K\left[\alpha \left \|\nabla F(\bar{\bm{z}}^{r,t})\right \|_2^2+ \beta \right]  \right\}\\\nonumber
	\overset{(b)}{\leq} & 
	\frac{6\eta^2E^2(E-1)}{1-D_E}\left\{\sum_{\mathcal{M}}K_0\delta_m^2 + K_0L^2\sum_{\mathcal{M}}\mathbb{E}\left[\left \|\bar{\bm{z}}_{m}^{r,t}- \bar{\bm{z}}^{r,t}\right \|_2^2\right]\right.\\
	+&\left. K\left[\alpha \left \|\nabla F(\bar{\bm{z}}^{r,t})\right \|_2^2+ \beta \right] \right\} + \frac{\eta^2\sigma^2E(E-1)K}{1-D_E},
\end{align}
where $(a)$ holds due to the fact that $e\leq E-1$, $ (b) $ holds by setting $D_E = 2\eta^2L^2E(E-1)$, and $\eta<\frac{1}{LE\sqrt{2T(T-1)}}<\frac{1}{L\sqrt{2E(E-1)}}\leq \frac{1}{2EL},T\geq2,E\geq2$.
Based on the results in \eqref{eq: T41}, we have the upper bound of $T_4$, which is
\begin{align}\nonumber\label{eq: T42}
	T_4
	\leq& \frac{\eta\sigma^2TD_E}{2(1-D_E)}+ \frac{3\eta TD_E}{1-D_E}\bar{\delta}^2+ \frac{3D_E}{M(1-D_E)}T_3\\
	&+\frac{3\eta ED_E}{(1-D_E)}\sum_{t=0}^{T-1}\left[\alpha \left \|\nabla F(\bar{\bm{z}}^{r,t})\right \|_2^2+ \beta \right],
\end{align}
where $\bar{\delta}^2 = \frac{1}{M}\sum_{\mathcal{M}}\delta_m^2$.
To summarize the results of \eqref{eq: T3'} and \eqref{eq: T42}, we have
\begin{align*}
	&\left\{	
	\begin{array}{l}
		T_3 \leq A_1 T_4 + B_1\\
		T_4 \leq A_2 T_3 + B_2
	\end{array}
	\right. 
	\Rightarrow 
	\begin{bmatrix}
		1	& -A_1\\
		-A_2	& 1
	\end{bmatrix}
	\begin{bmatrix}
		T_3	\\
		T_4
	\end{bmatrix}\preceq 
	\begin{bmatrix}
		B_1	\\
		B_2
	\end{bmatrix}
	\\
	&\Rightarrow 
	\begin{bmatrix}
		T_3	\\
		T_4
	\end{bmatrix}\preceq \frac{1}{1-A_1A_2}
	\begin{bmatrix}
		1	& A_1\\
		A_2	& 1
	\end{bmatrix}
	\begin{bmatrix}
		B_1	\\
		B_2
	\end{bmatrix},
\end{align*}
where $A_1 = \frac{D_T}{1-D_T}$, $A_2 = \frac{3D_E}{1-D_E}$, $B_1=\frac{\eta L^2(M-1)\sigma^2D_T}{K(1-D_T)}
+ \frac{\eta E D_T}{1-D_T}\sum_{t=0}^{T-1}\left[\alpha \left  \| \nabla F(\bar{\bm{z}}^{r,t}) \right \|_2^2 + \beta \right]$, $B_2 = \frac{\eta\sigma^2TD_E}{2(1-D_E)}+ \frac{3\eta TD_E}{1-D_E}\bar{\delta}^2 +\frac{3\eta ED_E}{1-D_E}\sum_{t=0}^{T-1}\left[\alpha \left \|\nabla F(\bar{\bm{z}}^{r,t})\right \|_2^2+ \beta \right]$,
and $A_1A_2\neq1$.
In fact, what we need is actually $T_3+T_4$, we have
\begin{align}
	T_3 + T_4 \leq \frac{1+A_2}{1-A_1A_2}B_1 + \frac{1+A_1}{1-A_1A_2} B_2,
\end{align}
for $E\geq1,T\geq1$. This is because $B_1 = B_2 = 0$ when $E = 1, T=1$.
Let $A_1\leq \frac{1}{3}$ and $A_2\leq \frac{3}{2}$, one can obtain that
\begin{align}\nonumber\label{eq: T3+T4}
	&T_3 + T_4 \\\nonumber
	\leq& \frac{1+A_2}{1-A_1A_2}B_1 + \frac{1+A_1}{1-A_1A_2} B_2\leq 5B_1 + \frac{8}{3}B_2\\\nonumber
	\leq&\frac{20}{3K}\eta L^2(M-1)\sigma^2D_T + \frac{20}{3}\eta E\alpha D_T \sum_{t=0}^{T-1} \left  \| \nabla F(\bar{\bm{z}}^{r,t}) \right \|_2^2 \\\nonumber
	+ & \frac{20}{3}\eta ET\beta D_T + 2\eta T \sigma^2D_E + 12\eta T \bar{\delta}^2D_E\\
	+&12\eta E\alpha D_E\sum_{t=0}^{T-1} \left  \| \nabla F(\bar{\bm{z}}^{r,t}) \right \|_2^2 + 12\eta ET\beta D_E,
\end{align}
with $\eta\leq\frac{1}{2EL\sqrt{2T(T-1)}}$, $\frac{1}{1-D_T}\leq\frac{4}{3}$, and $\frac{1}{1-D_E}\leq \frac{3}{2}$.
Substituting \eqref{eq: T3+T4} into \eqref{eq: expectation3}, we have
\begin{align}\nonumber\label{eq: expectation4}
	&\mathbb{E}\left[F(\bar{\bm{z}}^{r+1}) - F(\bar{\bm{z}}^{r})\right]\\\nonumber
	\leq&-\eta E\left[\frac{ 1}{2}-\alpha \left(\frac{20}{3}D_T+12D_E\right)\right]\sum_{t=0}^{T-1}\left \|\nabla  F(\bar{\bm{z}}^{r,t})\right \|_2^2\\\nonumber
	 + &\frac{20}{3K}\eta L^2(M-1)\sigma^2D_T + \frac{\eta^2EL\sigma^2T}{K}\\\nonumber
	+& 20\eta ET\beta D_T/3 + 2\eta T \sigma^2D_E + 12\eta T \bar{\delta}^2D_E + 12\eta ET\beta D_E\\\nonumber
	\overset{(a)}{\leq}&-\frac{\eta E}{4}\sum_{t=0}^{T-1}\left \|\nabla  F(\bar{\bm{z}}^{r,t})\right \|_2^2 + \frac{20}{3K}\eta L^2(M-1)\sigma^2D_T\\
	 +& 2\eta T D_E\left(\sigma^2 + 6 \bar{\delta}^2\right) + \frac{4}{3}\eta ET\beta (5D_T + 9D_E) +\frac{\eta^2EL\sigma^2T}{K},
\end{align}
where $ (a) $ holds by setting $ \frac{1}{2}-\alpha \left(\frac{20}{3}D_T+12D_E\right) \geq \frac{1}{4}$, which makes the learning rate $\eta$ satisfy 
$\eta \leq \min\left\{\frac{1}{8EL\sqrt{3\alpha T(T-1)}},\frac{1}{8L\sqrt{5/3\alpha E(E-1)}}\right\}.$

Finally, by taking averaging for \eqref{eq: expectation4} from $t=0$ to $T-1$ and $r = 0$ to $R-1$, we have the final results presented in the theorem.
        
	\apptocmd{\thebibliography}{\setlength{\itemsep}{-1pt}}{}{}	
	\bibliographystyle{IEEEtran}
	\bibliography{refs}

\end{document}